\documentclass[conference,onecolumn]{IEEEtran}

\IEEEoverridecommandlockouts                              
\usepackage{enumerate}

\usepackage{graphicx, epsfig}\usepackage{epstopdf}

\newtheorem{Remark}{Remark}

\newtheorem{Lemma}{Lemma}

\usepackage{algorithmic}
\usepackage{algorithm}

\usepackage{amsmath}
\usepackage{wasysym}

\usepackage{stfloats} 

\usepackage{url}

\usepackage{amsfonts}
\usepackage{amssymb}

\usepackage{multirow}

\makeatletter

\newcommand{\Rmnum}[1]{\expandafter\@slowromancap\romannumeral #1@}
\makeatother

\usepackage{authblk}
\title{\LARGE \bf On Achievable Schemes of Interference Alignment in Constant Channels via Finite Amplify-and-Forward Relays}
\author{Haichuan Zhou, Tharm Ratnarajah$\dag$\\
 The Institute of Electronics, Communications and Information Technology (ECIT),\\
Queen's University Belfast, Queen's Road, Belfast, UK\\
$\dag$The University of Edinburgh, Edinburgh, UK\\
Email: hzhou01@qub.ac.uk }

\begin{document}

\maketitle
\thispagestyle{empty}
\pagestyle{empty}

\begin {abstract}
This paper elaborates on the achievable schemes of interference
alignment in constant channels via finite amplify-and-forward (AF)
relays. Consider $K$ sources communicating with $K$ destinations
without direct links besides the relay connections. The total number
of relays is finite. The objective is to achieve interference
alignment for all user pairs to obtain half of their
interference-free degrees of freedom. In general, two strategies are
employed: \textit{coding at the edge} and \textit{coding in the
middle}, in which relays show different roles. The contributions are
that two fundamental and critical elements are captured to enable
interference alignment in this network: channel \textit{randomness}
or \textit{relativity}; subspace dimension suppression.
\end {abstract}

\IEEEpeerreviewmaketitle

\section{Introduction}


Since interference alignment (IA) is a new multiplexing gain
maximizing technique \cite{IA-DOF-Kuser-Interference} and relay is
considered a cost-effective solution for coverage extension and
capacity enhancement, recently many researchers have been
investigating schemes to combine these two techniques
\cite{can-increase-DOF,DOF-Multisource-Relay}. Various scenarios
have been intensively studied, e.g. decode-and-forward relays
\cite{interference-relay-dynamic-DF,relay-beamforming-Interference-Pricing-two-hop},
multi-user broadcast channel
\cite{linear-precoding-AF-MU-twoway-relay}, clustered relays
\cite{DoF-parallel-relay}, two-way relay selection
\cite{optimal-relay-selection-allocation-MU-analog-twoway}. Among
all the scenarios, the model of multi-user peer-peer two-hop
interference channel with amplify-and-forward relays without direct
links has received increasing attention due to its wide application
in practice
\cite{IA-MIMO-AF-IFC,IA-MU-TwoWay-Relay,achievable-two-hop-interference-conferencing-relay,IA-aided-relay-quasi-static-X},
which show that relays help the system to obtain higher degrees of
freedom (DoF) in high signal-to-noise ratio (SNR) regime.

While peer-peer multi-user interference channel has been
traditionally regarded as a challenging scenario due to channel
\textit{inseparability} \cite{Ergodic-IA}, not to mention with
relays in between. Furthermore, current relay networks are
categorized into two main sets: one is that relays are auxiliary
links in addition to direct links between both ends
\cite{can-increase-DOF,relay-IA-feasibility-algorithm,RelayAided-IA-QuasiStaticX,Relay-IA-QuasiStatic-IFC,achieve-DoF-K-user-MIMO-IFC-relay};
the other is that relays are the only connection path for end nodes
without direct links \cite{DOF-Multisource-Relay}. In general, relay
networks with direct links are well structured to construct IA.
However, this work considers relay networks without direct links
where relay generated equivalent channels are quite complicated and
poorly structured
\cite{DOF-Multisource-Relay,DoF-region-class-multisource-Gaussian-relay,aligh-Int-neutra-222}.
In that case, many approaches are proposed, e.g.
\cite{DOF-Multisource-Relay} requires an infinite number of relays
to eliminate interference;
\cite{DoF-region-class-multisource-Gaussian-relay} exploits ergodic
nature in fading channels; \cite{aligh-Int-neutra-222} illustrates
interference \textit{neutralization} in a special $2\times2\times2$
network; \cite{IA-MIMO-AF-IFC} uses mean squared error (MSE)
numerical method to minimize interference.

Three important features are highlighted in the scenario of this
work: the number of relays is finite so that it is impossible to
have the unpractical solution in \cite{DOF-Multisource-Relay} to
eliminate interference with infinite number of relays; each node
could only have single antenna to achieve IA; time-invariant
channels are also available to achieve IA even with single antenna
nodes \cite{IA-DOF-Kuser-Interference}. This work borrows the two
strategies generalized in multiple unicast problem
\cite{NC-Three-Unicast-IA}, which are coding \textit{at the edge}
and coding \textit{in the middle}. In the first strategy of coding
at the edge, relays randomly construct {equivalent} channels while
end nodes proceed with conventional interference alignment schemes.
Compared with other research, the most unfavorable conditions are
set in this work: all end and relay nodes are single antenna; all
channels are time-invariant; relay number is set to be finite, e.g
one or two. Max-flow-min-cut theorem is not directly applicable in
this scenario. In the second strategy of coding in the middle,
optimization techniques are applied to numerically approach
interference alignment with all nodes set to be multi-antenna.

Our contributions are in two folds:


A new fact is unveiled that when the network has only one relay,
Cadambe-Jafar scheme is not applicable in signal \textit{vector}
alignment; Motahari-Khandani scheme is not applicable in signal
\textit{level} alignment; asymmetric complex signaling is not
applicable in \textit{phase} alignment. The reason is that
relay-emulated channels lose {randomness} or {relativity}. Then the
problem is settled by two ways. On one hand, at least two relays are
necessary to emulate qualified channels to generate randomness. On
the other hand, to generate randomness still in the one-relay
channel, space-time type of precoding methods are applicable at the
edge; the conceptual idea of \textit{blind} interference alignment
is also applicable to {fluctuate} channels with the only one relay.
By all these analysis and schemes, a novel unique role of relays in
the network is revealed to overcome the channel randomness issue.



A novel solution is proposed for IA via relay coding based on a
non-trivial application of existing rank constrained rank
minimization (RCRM) method \cite{IA-RCRM,IA-RCRM-Globcom}.
Conventional optimization and numerical algorithms are non-convex or
hard for interference alignment
\cite{approaching-capacity-IA,IA-MIMO-AF-IFC,relay-beamforming-Interference-Pricing-two-hop}.
Since RCRM method is originally not for relay networks, the
convexity is proved for the new application. This novel solution has
three advantages: a) it considers DoF at high SNR by subspace
dimension minimization while other numerical methods could only
consider sum rate and mean squared error (MSE); b) it is universal
to be conveniently applied to a number of scenarios, e.g. MIMO
amplify-and-forward(AF) relay channel, MIMO two-way AF relay
channel, MIMO Y channel, and MIMO multi-hop relay channel. c)
actually conventional analytical and numerical solutions are hard to
obtain for the mentioned scenarios, while this novel method
accomplishes the design and is robust to poor conditions. Numerical
results show its effectiveness.

\section{Problem Statement and Model Description}\label{chater-relay-section-definition}
\label{sect:probdef}

\subsection{Basic Model} Consider $K$ sources and $K$ destinations connected by $L$
relays as shown in Fig. \ref{relay-IA}. Direct links between end
users are not available. All nodes are single antenna, denote the
channel from source $i$ to relay $l$ as $\mathrm{h}_{li}$ and the
channel from relay $l$ to destination $j$ as $\mathrm{f}_{jl}$.
Denote the sets $\mathcal{K}=\{1,2,\ldots,K\}$,
$\mathcal{L}=\{1,2,\ldots,L\}$, so that $i,j\in\mathcal{K}$,
$l\in\mathcal{L}$. All channels are quasi-static, i.e.
$\mathrm{h}_{li}$ and $\mathrm{f}_{jl}$ are scalar constants.

\begin{figure}[htpb]
  \centering
    \includegraphics[width=3.0in]{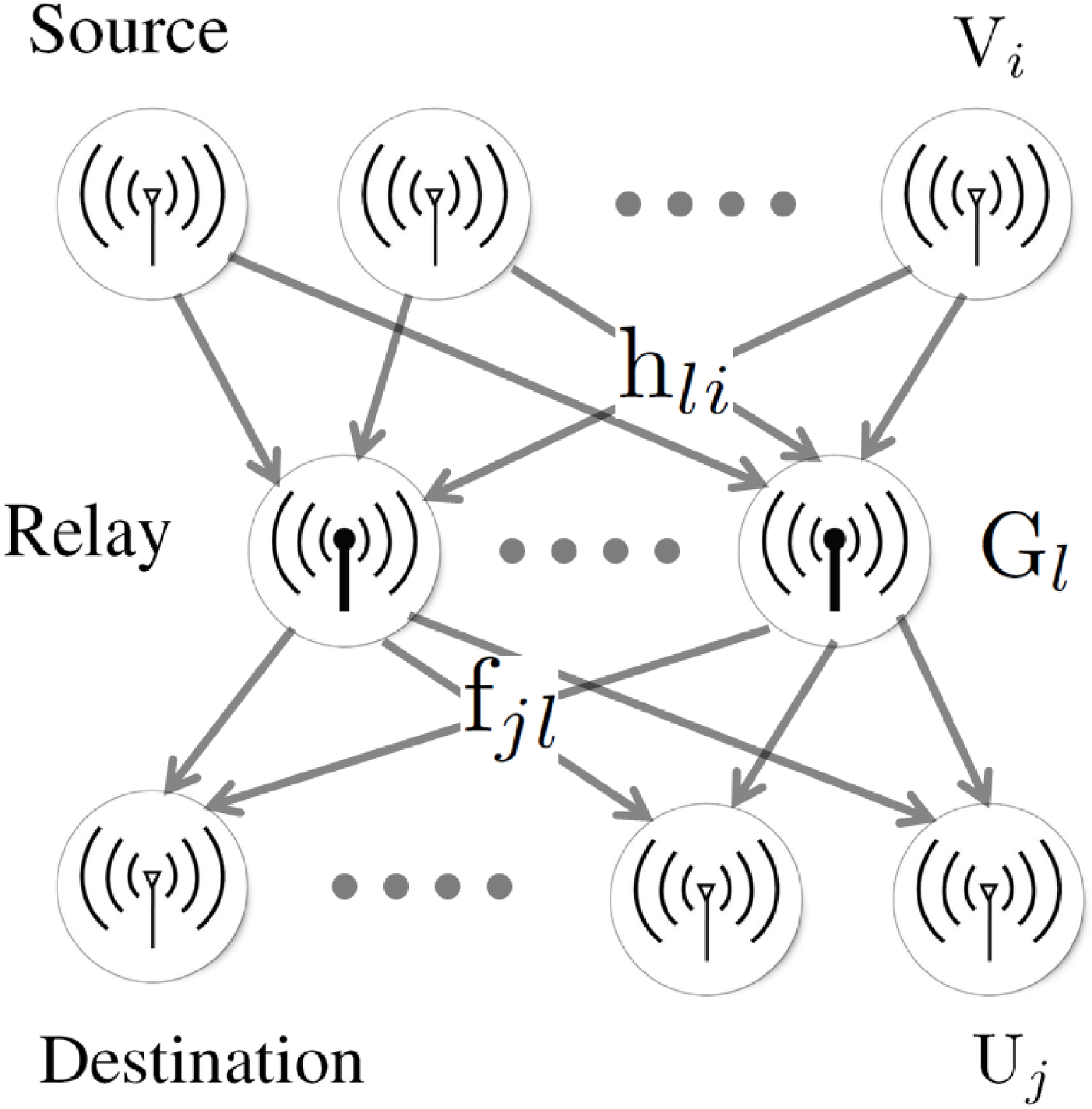}\\
  \caption{Basic model of sources, relays, and destinations}\label{relay-IA}
\end{figure}

Time-extended MIMO scheme in \cite{DOF-Multisource-Relay} is used so
that $T$ consecutive symbols form one signal and the relay coding
matrix at $l$-th relay is
\begin{equation}\label{G}
\begin{aligned}
\mathrm{G}_{l}=\Big[\mathrm{g}_{l}(p,q)\Big]_{T\times T}
\end{aligned}
\end{equation}
in which $\mathrm{g}_{l}(p,q)$ is the relay gain from the $q$-th
time slot to the $p$-the time slot at the $l$-th relay. Then the
received signal at $j$-th user can be written as


\begin{equation}\label{output-HD-bf}
\begin{aligned}
\mathrm{Y}_{j}&=\Theta_{jj}\mathrm{V}_j\mathrm{X}_j+\sum_{i=1,i\neq
j}^{K}\Theta_{ji}\mathrm{V}_i\mathrm{X}_i+\mathrm{Z}_{j}
\end{aligned}
\end{equation}

\begin{equation}\label{Theta}
\begin{aligned}
\Theta_{ji}=\sum_{l=1}^{L}\mathrm{f}_{jl}\mathrm{G}_{l}\mathrm{h}_{li}
\end{aligned}
\end{equation}

where $\mathrm{X}_{i}\in\mathbb{C}^{d\times 1}$ is the $d$ input
streams at $i$-th source and $\mathrm{V}_{i}\in\mathbb{C}^{T\times
d}$ is precoding matrix. $\mathrm{Z}_{j} \in\mathbb{C}^{T\times 1}$
is the noise at $j$-th receiver, and
$\mathrm{Y}_{j}\in\mathbb{C}^{T\times 1}$. The equivalent channel
matrix from the $i$-th source to the $j$-th destination is defined
as $\Theta_{ji}$ for simplicity. The holistic system equation is
also written as:

\begin{equation}\label{output-HD-SYS}
\left[
\begin{aligned}
&\mathrm{Y}_{1}\\
&\mathrm{Y}_{2}\\
&\ \ \vdots \\
&\mathrm{Y}_{K}
\end{aligned}
\right]
=
\sum_{l=1}^{L}
\left[
\begin{aligned}
&\mathrm{f}_{1l}\mathrm{h}_{l1}\mathrm{G}_{l}\cdots\mathrm{f}_{1l}\mathrm{h}_{lK}\mathrm{G}_{l}\\
&\mathrm{f}_{2l}\mathrm{h}_{l1}\mathrm{G}_{l}\cdots\mathrm{f}_{2l}\mathrm{h}_{lK}\mathrm{G}_{l}\\
&\ \ \ \ddots\ \ \ \mathrm{f}_{jl}\mathrm{h}_{li}\mathrm{G}_{l}\ \ \ddots\\
&\mathrm{f}_{Kl}\mathrm{h}_{l1}\mathrm{G}_{l}\cdots\mathrm{f}_{Kl}\mathrm{h}_{lK}\mathrm{G}_{l}
\end{aligned}
\right]
\left[
\begin{aligned}
&\mathrm{V}_1\mathrm{X}_1\\
&\mathrm{V}_2\mathrm{X}_2\\
&\ \ \ \vdots\\
&\mathrm{V}_K\mathrm{X}_K
\end{aligned}
\right]
+
\left[
\begin{aligned}
&\mathrm{Z}_{1}\\
&\mathrm{Z}_{2}\\
&\ \ \vdots \\
&\mathrm{Z}_{K}
\end{aligned}
\right]
\end{equation}

\subsection{General Model}

For more general cases, the model is illustrated in Fig.
\ref{relay-IA-pic2}. Denote $K$ nodes on one end as ${I}_1, {I}_2,
\cdots, {I}_K$; $L$ relay nodes in the middle as ${R}_1, {R}_2,
\cdots, {R}_L$; $K$ nodes on the other end as ${J}_1, \cdots,
{J}_{K-1}, {J}_K$. Define sets $\mathcal{K}=\{1,2,\ldots,K\}$,
$\mathcal{L}=\{1,2,\ldots,L\}$. All channels are {time-invariant}.
Assume all nodes have all channel knowledge to cooperate.

\vspace{-0mm}
\begin{figure}[htpb]
  \begin{center}
    \includegraphics[width=3.0in]{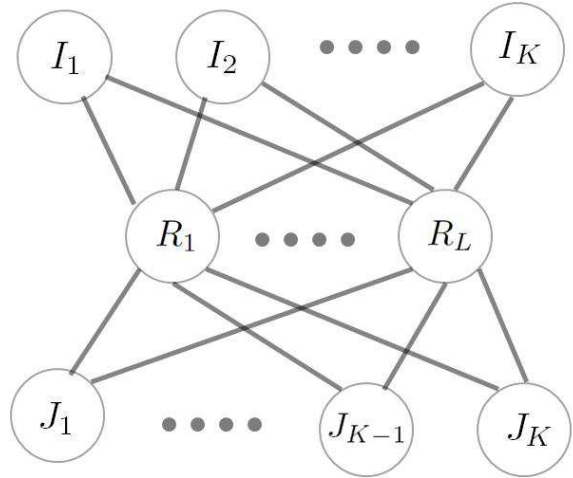}\\
  \end{center}
  \caption{General model of sources, relays, destinations}\label{relay-IA-pic2}
\end{figure}

In the case when all nodes are single antenna, denote
$\mathrm{h}_{R_lI_i}$, $\mathrm{h}_{J_jR_l}$, as channels from $I_i$
to $R_l$, $R_l$ to $J_j$, respectively, where $i,j\in\mathcal{K}$,
$l\in\mathcal{L}$, all channel coefficients are scalars. While in
the case when all nodes are multiple antenna, denote
$\mathrm{H}_{R_lI_i}$, $\mathrm{H}_{I_iR_l}$, $\mathrm{H}_{J_jR_l}$,
$\mathrm{H}_{R_lJ_j}$, as channels from $I_i$ to $R_l$, $R_l$ to
$I_i$, $R_l$ to $J_j$, $J_j$ to $R_l$, respectively, where
$i,j\in\mathcal{K}$, $l\in\mathcal{L}$, all channel coefficients are
matrices.
%
%
%

\section{Characteristics of Amplify-and-Forward and Finite Relays}

Before detailed design and analysis, this relay scenario needs to be
characterized and some critical features should be exposed. First,
the applicability of max-flow-min-cut theorem in this network is
actually in doubt along with the amplify-and-forward scheme. Second,
the number of relays, achievable DoF, and coding strategy are
correlated in a complicated form.

\subsection{Amplify-and-Forward Scheme and Max-Flow-Min-Cut Theorem}

A natural question arises here: Does max-flow-min-cut theorem still
work in this IA relay network? Another similar question is: in a
$K$-pair network with finite relays, what does the achievable DoF
equal to? The answers to these two questions are quite nuanced
actually.

Look into the $K$-pair network with $L$ relays. The first hop is a
$K\times L$ MIMO link and the second hop is a $L\times K$ MIMO link.
Conventionally, both the first and the second channel matrices only
have a rank of $L$, which means that for the whole $K$-pair network,
the minimum cut is $L$. However, the degrees of freedom for the
network is not necessarily $L$. Consider the $K$-pair network
regardless of the middle connections, so that there are potentially
$K$ DoF for this equivalent network as long as with joint processing
across the $K$ transmitters or receivers. The critical difference is
that the intermediate relay nodes do not necessarily process data
flows directly, and could act only as equivalent channels. In
addition, instead of joint processing, a distributed manner such as
interference alignment could still possibly achieve $K/2$ DoF.
However, it is also important to highlight that when there are
finite single-antenna relays, the equivalent channel may be rank
deficient, which is discussed detailedly in the following parts.

DoF as in the signal processing model is not simply equal to
capacity flow links as in network coding model. In the flow network
model, each node has separate inputs and outputs, and each edge
represents a separate link. While in amplify-and-forward relay
networks, arbitrary flows could overlap through same relays, and
links in each hop does not
represent separate flows.
Max-flow min-cut theorem is originally applied to wired networks
with single-letter characterizations. A recent model for wireless
relay networks is known as the linear deterministic relay network
model with a max-flow min-cut result pertaining to it
\cite{deterministic-relay}. The algorithmic framework is introduced
by Avestimehr, Diggavi and Tse, incorporating the key features of
broadcasting and superposition. The signals are elements of a finite
field and the interactions between the signals are assumed to be
linear. This model is based on linking systems and the max-flow
min-cut theorem is applicable with matroid intersection or
partition.

Moreover, \cite{NC-Three-Unicast-IA} shows that in a network with 3
unicast sessions each with min-cut of 1, whenever network alignment
can achieve rate of 1/2 per session, there exists an alternative
approach including routing, packing butterflies, random linear
network coding, or other network coding strategies instead of
alignment. However when there are more than 3 sessions, alignment is
required to obtain the maximum rate as 1/2 the min-cut, while no
other method can achieve it.

In summary, the max-flow min-cut theorem is not exactly available
for this amplify-and-forward system, and DoF is not determined
directly based on the number of relays. So that it is non-trivial to
investigate new schemes in the amplify-and-forward strategy. Before
further understanding and analyzing the DoF of this network,
different situations need to be classified as following.

\subsection{Finite/Infinite Relays, DoF Limits, and Strategies}

For the $K$-pair amplify-and-forward relay network, all the cases
are roughly classified into three categories: when infinite relays
are provided, full $K$ DoF is achievable, by using relay coding;
when finite relays are provided, only fractional DoF is achievable
as the number of users $K$ grows, by using conventional MIMO
precoding approach; when a specific range of finite relays are
provided, $K/2$ DoF is possible to be obtained, by using asymptotic
interference alignment precoding approach. The details are shown in
Table \ref{AF-3-Classification} and discussed as following.

\begin{table}[!hbp]
\centering
\begin{tabular}{|c|c|c|c|c|c|}
\hline
Case &  Relay Number  &  Target DoF  &  Coding  &  Channel  & IA Scheme  \\
\hline
  1  &  Infinite  &  $K$  &  Relay Coding  &  Generic  &   N/A   \\
\hline
  2  &  Finite  &  $\frac{2}{K+1}$  &  Precoding  &  $\begin{aligned}&\ \ \text{Generic}\\&\text{or Diagonal}\end{aligned}$  &   $\begin{aligned}&\ \ \ \text{Leakage}\\&\text{Minimization}\end{aligned}$   \\
\hline
  3  &  $\begin{aligned}&\text{Specific}\\&\text{Finite}\end{aligned}$  &  $K/2$  &  Precoding  &  Diagonal  &   Asymptotic   \\
\hline
\end{tabular}\caption{Classifications according to relay number, DoF Limits, and Strategies}\label{AF-3-Classification}
\end{table}

\subsubsection{Infinite Relays, Full DoF, and Relay Coding}

Consider the single antenna $K$-pair relay network in Fig.
\ref{relay-IA} with the system equation of (\ref{output-HD-bf}) and
the equivalent channel in (\ref{Theta}). Let the interference at any
receiver to be zero, and then the condition for $j$-th user should
satisfy the following:

\begin{equation}\label{dummy-basis-2}
\begin{aligned}
&\ \ \ \ \ \ \mathrm{\Theta}_{ji}=\mathbf{0}\ \ \ \ \forall i\neq j\\
&\text{i.e.}\ \
\sum_{l=1}^{L}\mathrm{f}_{jl}\mathrm{G}_{l}\mathrm{h}_{li}=\mathbf{0}
\ \ \ \forall i\neq j
\end{aligned}
\end{equation}

Equation (\ref{dummy-basis-2}) represents $K(K-1)$ matrix equations,
equalling to $T^2K(K-1)$ linear equations. The matrices
$\mathrm{G}_{1},\mathrm{G}_{2},\cdots,\mathrm{G}_{L}$ contain $LT^2$
variables. Generally, the equations are solvable when: $LT^2\geq
T^2K(K-1)$, i.e. $L\geq K(K-1)$. It could be also interpreted as for
each set of $L$ elements of
$\{\mathrm{G}_{1}(l_1,l_2),\mathrm{G}_{2}(l_1,l_2),\cdots,\mathrm{G}_{L}(l_1,l_2)\}$
there are $K(K-1)$ equations. Therefore in this case, the network
achieves $K$ DoF by only using the linear relay coding.

%
%
%
%
%

\subsubsection{Finite Relays, Fractional DoF, and Precoding}

Consider the same network in Fig. \ref{relay-IA}. However, there're
only finite relays in this case. Instead of relay coding as above, a
scheme of precoding on the equivalent channels is proposed, i.e.
following the \textit{coding at the edge} strategy
\cite{NC-Three-Unicast-IA}. Relays generate equivalent channels, and
the $K$ transmitters and $K$ receivers only see the equivalent
channels regardless of relays. Then conventional interference
alignment methods could be used directly for all $K$ users
\cite{IA-DOF-Kuser-Interference}, \cite{approaching-capacity-IA}.

The relay gain matrices of (\ref{G}) are randomly chosen to be in a
generic form (full elements), so that $\mathrm{G}_l$ and
corresponding $\mathrm{\Theta}_{ji}$ are equivalent to general MIMO
channels. The leakage minimization algorithm of
\cite{approaching-capacity-IA} is applicable to approach
interference alignment. However, according to the feasibility
condition in \cite{feasibility-IA-MIMO-IFC} for symmetric MIMO
channels, the equivalent system in (\ref{output-HD-bf}) must satisfy
\normalsize
\begin{equation}
T+T-d(K+1)\geq 0\label{MIMO-cond}
\end{equation}
\normalsize

So that each user obtains DoF bounded by
$\frac{d}{T}=\frac{2}{K+1}$. Numerical results are shown in Fig.
\ref{ICC-Fig-4} in which three cases denoted by `Generic' use the
leakage minimization algorithm in \cite{approaching-capacity-IA} to
achieve IA for different numbers of relays: relay number $L=1$,
$L=5$ and no relay. Total user number $K=5$, and the time extension
length $T=21$, and each user has $d=5$ streams.A MIMO network with
21 antennas for all nodes is set as a reference case. When the
channels are generated by 1 relay or 5 relays, each user could
obtain the same number of DoF as the reference MIMO channel case
roughly as $2/(K+1)$.

\subsubsection{Specific Finite Relays, Half DoF, and Precoding}

Although the feasibility condition in \cite{feasibility-IA-MIMO-IFC}
limits the DoF in the constant MIMO channels, however there is still
chance to achieve $K/2$ DoF for the network actually. The reason is
that the relays are capable of flexibly constructing desired channel
structures such as time-variant channels.

Manipulate relay gain matrices to be of diagonal structures as
following:

\normalsize
\begin{equation}\label{G-diag}
\begin{aligned}
\mathrm{G}_{l}=\mathrm{Diag}\{\mathrm{g}_{l}(1,1),
\mathrm{g}_{l}(2,2), \cdots, \mathrm{g}_{l}(T,T)\}
\end{aligned}
\end{equation}
\normalsize

where $\mathrm{Diag}\{\cdot\}$ is diagonal function which place all
inputs on diagonal line of a matrix output. Then the equivalent
channels are constructed in diagonal structure as well:

\normalsize
\begin{equation}\label{Theta-diag}
\begin{aligned}
\mathrm{\Theta}_{ji}=\sum_{l=1}^{L}\mathrm{Diag}\{\mathrm{f}_{jl}{\mathrm{g}_{l}}(1,1)\mathrm{h}_{li},
\mathrm{f}_{jl}\mathrm{g}_{l}(2,2)\mathrm{h}_{li}, \cdots,
\mathrm{f}_{jl}\mathrm{g}_{l}(T,T)\mathrm{h}_{li}\}
\end{aligned}
\end{equation}
\normalsize

If we only deal with the relay constructed diagonal channels by
using the conventional leakage minimization method as before, the
DoF results would have no improvement as shown in Fig.
\ref{ICC-Fig-4}. It shows three cases denoted by `Diagonal', which
indicate the same $2/(K+1)$ DoF curves as the equivalent generic
MIMO channels where $K=5, T=21, d=5$. Observe that the relay number
does not affect the algorithmic IA result either. Theoretically,
there has been no conclusions so far on algorithmic IA feasibility
of diagonal channels \cite{feasibility-IA-MIMO-IFC}.

\begin{figure}[htpb]
  \centering
    \includegraphics[width=4.5in]{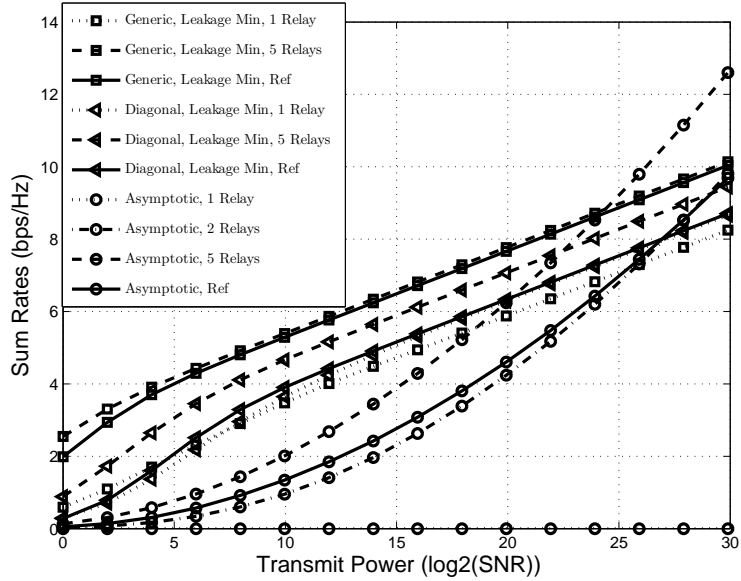}\\
  \caption{Coding on the edge with leakage minimization and asymptotic design, with different number of relays}\label{ICC-Fig-4}
\end{figure}

However, if we deal with the relay constructed diagonal channels by
using the analytical asymptotic design (Cadambe-Jafar scheme) for
frequency selective channels \cite{IA-DOF-Kuser-Interference}, each
user could approach $1/2$ DoF regardless of the number of $K$ users
in the network. Take the 3-user network as an example to apply
Cadambe-Jafar scheme in this relay network. The system equation
(\ref{output-HD-bf}) needs to be modified a little: let \normalsize
${\mathrm{{V}}_{1}}\in\mathbb{C}^{T\times (n+1)}$,
${\mathrm{{V}}_{2}}\in\mathbb{C}^{T\times n}$,
${\mathrm{{V}}_{3}}\in\mathbb{C}^{T\times n}$, $T=2n+1 \text{
}$\normalsize, where $n$ is a positive integer. Thus DoF for the
three users are not symmetric here and set as $\frac{n+1}{2n+1}$,
$\frac{n}{2n+1}$, $\frac{n}{2n+1}$ respectively. When $n$ is large
enough, each user could approach $1/2$ DoF. Design the aligned
interference subspace at each user as following:

\normalsize
\begin{equation}\label{time-extension-channel-ref}
\begin{aligned}
&\mathrm{\Theta}_{12}\mathrm{V}_{2}=\mathrm{\Theta}_{13}\mathrm{{V}}_{3},\mathrm{\Theta}_{23}\mathrm{V}_{3}\prec\mathrm{\Theta}_{21}\mathrm{{V}}_{1},
\mathrm{\Theta}_{32}\mathrm{V}_{2}\prec\mathrm{\Theta}_{31}\mathrm{{V}}_{1}
\end{aligned}
\end{equation}
\normalsize

Then \cite{IA-DOF-Kuser-Interference} proposed the following
analytical design for each precoding matrix to satisfy
(\ref{time-extension-channel-ref}):
\begin{equation}\label{time-extension-IA}
\begin{aligned}
\mathrm{{A}}&=\mathrm{\Theta}_{12}\mathrm{\Theta}_{21}^{-1}\mathrm{\Theta}_{23}\mathrm{\Theta}_{32}^{-1}\mathrm{\Theta}_{31}\mathrm{\Theta}_{13}^{-1}\\
\mathrm{{V}}_{1}&=[\mathrm{w} \ \mathrm{w}\mathrm{A} \
\mathrm{w}\mathrm{A}^2 \ \cdots \ \mathrm{w}\mathrm{A}^n]\\
\mathrm{{V}}_{3}&=\mathrm{\Theta}_{21}\mathrm{\Theta}_{23}^{-1}[\mathrm{w}\mathrm{A}
\
\mathrm{w}\mathrm{A}^2 \ \cdots \ \mathrm{w}\mathrm{A}^n]\\
\mathrm{{V}}_{2}&=\mathrm{\Theta}_{31}\mathrm{\Theta}_{32}^{-1}[\mathrm{w}
\
\mathrm{w}\mathrm{A} \ \cdots \ \mathrm{w}\mathrm{A}^{n-1}]\\
\end{aligned}
\end{equation}

where $\mathrm{w}=[1 \ 1 \ \cdots \
1]^{\dag}\in\mathbb{C}^{(2n+1)\times 1}$. In this way, the
three-user network could approach $3/2$ DoF eventually. It is
important to highlight that this scheme could be extended to
arbitrary $K$-user case. Then (\ref{time-extension-channel-ref}) and
(\ref{time-extension-IA}) are upgraded to a more complicated form
accordingly \cite{IA-DOF-Kuser-Interference}. The achievable
multiplexing gain of the network becomes
$\frac{(n+1)^N+(K-1)n^N}{(n+1)^N+n^N}$, where
\small$N=(K-1)(K-2)-1$\normalsize, $n$ is non-negative integer. As a
result each user could approach $1/2$ DoF.


Numerical results are shown in Fig. \ref{ICC-Fig-4}. In the 3-user
network, let $n=10$ so that the time extension length $T=21$. There
are four cases: relay number $L=1$, $L=2$, $L=5$, and no relay. In
addition, a case of frequency selective channel with 21 dimensions
for all nodes is introduced as a reference. The four cases using the
asymptotic design are denoted by `Asymptotic'. When the channels are
generated by 2 relays or 5 relays, each user could obtain the same
DoF as the reference frequency selective channel case. However, when
the channels are generated by only 1 relay, each user obtains zero
rate, and zero DoF, i.e. the interference alignment scheme fails.

The reason for the failure of the case of $L=1$ is investigated as
following. For $L=1$ relay, (\ref{Theta}) becomes:

\normalsize
\begin{equation}\label{Theta-SR}
\Theta_{ji}=f_{j1}\mathrm{G}_1h_{1i}
\end{equation}
\normalsize

Notice $\mathrm{G}_l$ is diagonal as in (\ref{G-diag}), then the
asymptotic design in (\ref{time-extension-IA}) degenerates to the
following form:

\normalsize
\begin{equation}\label{time-extension-IA-SR}
\begin{aligned}
\alpha&=f_{11}h_{12}(f_{21}h_{11})^{-1}f_{21}h_{13}(f_{31}h_{12})^{-1}f_{31}h_{11}(f_{11}h_{13})^{-1}\\
\mathrm{{A}}&=\alpha\mathrm{I}_T\\
\mathrm{{V}}_{1}&=[\mathrm{w} \ \alpha\mathrm{w} \
\alpha^2\mathrm{w} \ \cdots \ \alpha^n\mathrm{w}]\\
\mathrm{{V}}_{3}&=f_{21}h_{11}(f_{21}h_{13})^{-1}[\alpha\mathrm{w} \
\alpha^2\mathrm{w} \ \cdots \ \alpha^n\mathrm{w}]\\
\mathrm{{V}}_{2}&=f_{31}h_{11}(f_{31}h_{12})^{-1}[\mathrm{w} \
\alpha\mathrm{w} \ \cdots \ \alpha^{n-1}\mathrm{w}]\\
\end{aligned}
\end{equation}
\normalsize

where $\alpha$ is scalar coefficient equal to 1, and $\mathrm{I}_T$
is an identity matrix. Observe in (\ref{time-extension-IA-SR}) the
three precoding matrices $\mathrm{{V}}_{1}$, $\mathrm{{V}}_{2}$,
$\mathrm{{V}}_{3}$ have only rank 1 because all $\Theta_{ji}$ are
linear dependent to be eliminated, so that the asymptotic solution
fails due to the lack of \textit{channel randomness} or
\textit{relativity} in \cite{IA-DOF-Kuser-Interference}, which is
the core idea of the IA mechanism. Similarly, in the $K$-user case,
the asymptotic solution also fails when there is only 1 relay.
Meanwhile recall the numerical result by leakage minimization
algorithm in equivalent MIMO scheme, the achieved DoF is low but
non-zero.

Comparing with the case of $L=2$, the relay generated diagonal
channel in (\ref{Theta-diag}) becomes:
\begin{equation}\label{Theta-DR}
\Theta_{ji}=f_{j1}\mathrm{G}_1h_{1i}+f_{j2}\mathrm{G}_2h_{2i}
\end{equation}

It needs to be proved that $\Theta_{ji}$ are not linear dependent to
lose channel randomness with the following lemma.

\begin{Lemma}\label{comb}
$\Theta_{ji}$ and $\Theta_{nm}$ are linear independent almost surely
for arbitrary non-identical $(j,i)$ and $(n,m)$.
\end{Lemma}

\begin{proof} Assume

$\exists \beta_1, \beta_2\in\mathbb{C}$. $\beta_1, \beta_2\neq 0$,
such that $\beta_1\Theta_{ji}+\beta_2\Theta_{nm}=\mathrm{0}$.

Expand and group by $\mathrm{G}_1$ and $\mathrm{G}_2$ to:

\normalsize
$(\beta_1f_{j1}h_{1i}+\beta_2f_{n1}h_{1m})\mathrm{G}_1+(\beta_1f_{j2}h_{2i}+\beta_2f_{n2}h_{2m})\mathrm{G}_2=\mathrm{0}$
.\normalsize

$\mathrm{G}_1$ and $\mathrm{G}_2$ are arbitrary generic diagonal
matrices, therefore almost surely the coefficients are zeros, i.e.

\normalsize
$(\beta_1f_{j1}h_{1i}+\beta_2f_{n1}h_{1m})=0,(\beta_1f_{j2}h_{2i}+\beta_2f_{n2}h_{2m})=0$.\normalsize

Forming the ratio $-\frac{\beta_1}{\beta_1}$, so we have
$\frac{f_{j1}h_{1i}}{f_{n1}h_{1m}}=\frac{f_{j2}h_{2i}}{f_{n2}h_{2m}}$.

Since $f_{j1}, h_{1i}, f_{n1}, h_{1m}, f_{j2}, h_{2i}, f_{n2},
h_{2m}$ are all random independent scalars, the equality fails
almost surely. So that $\beta_1$, $\beta_2$ do not exist almost
surely. Then $\Theta_{ji}$ and $\Theta_{nm}$ are linearly
independent almost surely.
\end{proof}

Notice `almost surely' cases are recognized as successful
interference alignment \cite{IA-DOF-Kuser-Interference}. Lemma
\ref{comb} shows that
$\Theta_{12}\Theta_{21}^{-1}$,$\Theta_{23}\Theta_{32}^{-1}$, and
$\Theta_{31}\Theta_{13}^{-1}$ are not scaled identity matrices,
therefore their total product, $\mathrm A$, is not a scaled identity
matrix almost surely considering channel \textit{relativity} of
these three terms. As a result the asymptotic design in
(\ref{time-extension-IA}) still works here, rather than degenerating
to the form (\ref{time-extension-IA-SR}). In summary, Lemma
\ref{comb} reveals that two or more relays can generate
\textsl{channel randomness} or \textsl{channel relativity} in the
equivalent channels with problems in (\ref{Theta-SR}) and
(\ref{time-extension-IA-SR}) successfully prevented. Then the
asymptotic design in (\ref{time-extension-IA}) can be applied
successfully.

\section{Role of Relay in Precoding Scheme in Constant Channels}

The above 3-user case reveals a critical issue for the constant
channel to implement interference alignment. Therefore it is
important to look into general $K$-user networks, as well as
different class of IA schemes. In this section, for the general
scenarios of relay networks, the model is shown in Fig.
\ref{relay-IA-pic2} : define nodes ${I}_1, {I}_2, \cdots, {I}_K$ as
sources and ${J}_1, \cdots, {J}_{K-1}, {J}_K$ as destinations,
whereas each $I_k$ communicates with each ${J}_k$ via relays ${R}_1,
{R}_2, \cdots, {R}_L$. Relays are half-duplex so that the
transmission procedure consists of two stages where relays forward
data in the second stage. Interference alignment is designed with
the following strategy: relays construct equivalent channels; source
and destination nodes proceed precoding and zero-forcing, i.e.
\textit{coding at the edge} as \cite{NC-Three-Unicast-IA} claimed.
The objective is still to approach 1/2 due DoF (exclude duplex
factor) for every user, which is quite difficult and non-trivial for
a relay network under quasi-static channel condition. The work is
published in \cite{single-antenna-IA-finite-relay} and submitted in
\cite{achieve-IA-constant-finite-AF-relay}.

\subsection{Time-extended Signal Vector Alignment}

The most common approach to implement IA in this network is the
time-extended MIMO scheme as in \cite{DOF-Multisource-Relay}, where
$T$ consecutive symbols form one signal vector and IA is designed in
the vector space. The above example of 3-user network as in
equations (\ref{G-diag}) and (\ref{Theta-diag}) exactly constructs
the relay gain matrix $\mathrm{G}_{l}$ and equivalent channels
$\Theta_{ji}$ with this time-extended MIMO approach. Therefore, it
is important to extend the design to arbitrary $K$-user networks,
and look into the same issue caused by constant channels with single
relay in the network.


$\mathrm{G}_{l}$ and $\Theta_{ji}$ could be diagonal or generic. If
they are generic, it is equivalent to MIMO channel constrained by
feasibility conditions \cite{feasibility-IA-MIMO-IFC}, so that each
user could only approach $2/(K+1)$ DoF. If they are diagonal, the
channel is equivalent to a frequency selective channel in
\cite{IA-DOF-Kuser-Interference,beamforming-efficient-IA}, then the
asymptotic design of Cadambe-Jafar scheme could be applied. Notice
\cite{beamforming-efficient-IA} has the same core design structure
as \cite{IA-DOF-Kuser-Interference}, so that the design of
\cite{beamforming-efficient-IA} is illustrated here for a general
$K$-pair relay network as in the following form:

\normalsize
\begin{equation}\label{time-extension-IA-Kuser-a}
\begin{aligned}
\mathrm{{\Phi}}_{ji}&=\mathrm{\Theta}_{j1}^{-1}\mathrm{\Theta}_{ji}\mathrm{\Theta}_{1i}^{-1}\mathrm{\Theta}_{13}\hspace{3mm}\forall j,i\in\mathcal{K}\setminus\{1\}\\
\end{aligned}
\end{equation}
\begin{equation}\label{time-extension-IA-Kuser-b}
\begin{aligned}
\mathrm{{V}}_{3}&=\left\{\mathrm{{\Phi}}_{23}^{-1}\prod_{(j,i)\in\mathcal{A}}\left(\mathrm{{\Phi}}_{23}^{-1}\mathrm{{\Phi}}_{ji}\right)^{n_{ji}}\cdot\mathbf{1}_T\ \Big|\sum_{(j,i)\in\mathcal{A}}n_{ji}\leq n^*\right\}\\
\mathrm{{V}}_{1}&=\left\{\prod_{(j,i)\in\mathcal{A}}\left(\mathrm{{\Phi}}_{23}^{-1}\mathrm{{\Phi}}_{ji}\right)^{n_{ji}}\cdot\mathbf{1}_T\ \Big|\sum_{(j,i)\in\mathcal{A}}n_{ji}\leq n^*+1\right\}\\
\mathrm{{V}}_{i}&=\mathrm{\Theta}_{1i}^{-1}\mathrm{\Theta}_{13}\mathrm{V_3}\hspace{3mm}\forall i\in\mathcal{K}\setminus\{1,3\}\\
\end{aligned}
\end{equation}
\normalsize

where $\mathrm{V}_{i}\in\mathbb{C}^{T\times d_i}$ is the precoding
matrix for arbitrary source $I_i, i\neq 1,3$, which contains $d_i$
input streams. $\mathcal{A}=\left\{(j,i)\ |\
j,i\in\mathcal{K}\setminus\{1\},j\neq i,(j,i)\neq(2,3)\right\}$.
$\mathbf{1}_T\in\mathbb{C}^{T\times1}$ is the all one vector, $n^*$
is integer. Notice $T=\binom {n^*+N+1} N + \binom {n^*+N} N$ and
\small$N=(K-1)(K-2)-1$\normalsize. $\mathrm{{V}}_{1}$ has a
dimension of $T\times\binom {n^*+N+1} N$; $\mathrm{{V}}_{i}$ has a
dimension of $T\times\binom {n^*+N} N$, $\forall i\neq 1$. So that
each user could approach $1/2$ DoF when $n^*$ is large for arbitrary
$K$.

\subsubsection{The case when the network has single relay, $L=1$}

However, in the special case when $L=1$, this scheme does not work.
The reason is easily shown in the equations (\ref{Theta}) and
(\ref{time-extension-IA-Kuser-a}) as following:

\normalsize
\begin{equation}
\begin{aligned}\label{Phi-1R}
&\mathrm{{\Phi}}_{ji}\cdot\mathbf{1}_T=(\mathrm{h}_{J_jR_1}\mathrm{G}_{1}\mathrm{h}_{R_1I_1})^{-1}(\mathrm{h}_{J_jR_1}\mathrm{G}_{1}\mathrm{h}_{R_1I_i})\\
&\hspace{14mm}\cdot(\mathrm{h}_{J_1R_1}\mathrm{G}_{1}\mathrm{h}_{R_1I_1})^{-1}(\mathrm{h}_{J_1R_1}\mathrm{G}_{1}\mathrm{h}_{R_1I_3})\cdot\mathbf{1}_T\\
&\hspace{5mm}=\mathrm{h}_{J_jR_1}^{-1}\mathrm{h}_{R_1I_1}^{-1}\mathrm{h}_{J_jR_1}\mathrm{h}_{R_1I_i}\mathrm{h}_{J_1R_1}^{-1}\mathrm{h}_{R_1I_1}^{-1}\mathrm{h}_{J_1R_1}\mathrm{h}_{R_1I_3}\cdot\mathbf{1}_T
\end{aligned}
\end{equation}
\normalsize

Apply (\ref{Phi-1R}) to the precoding solutions of
(\ref{time-extension-IA-Kuser-b}), then we could observe column
subspace of $\mathrm{{V}}_{i}, \forall i\in\mathcal{K}$ in
(\ref{time-extension-IA-Kuser-b}) collapses to one dimensional space
parallel to $\mathbf{1}_T$.
It means the signal subspace of $\mathrm{{V}}_{i}$ degenerates so
that all the transmissions fail to obtain DoF. The following remark
is summarized.

\begin{Remark}
For the equivalent channels generated by single relay, interference
alignment is not feasible to approach 1/2 DoF for every user by
using existing asymptotic designs of
\cite{beamforming-efficient-IA,IA-DOF-Kuser-Interference}.
\end{Remark}

\subsubsection{The case when the network has at least two relays, $L\geq2$}

While for the case of $L\geq2$, asymptotic IA scheme could be much
likely applied to obtain $K/2$ DoF of this network for arbitrary
$K$. As a primitive investigation, set $L=2$. The key is to prove
that the column subspace of $\mathrm{{V}}_{i}, \forall
i\in\mathcal{K}$ in (\ref{time-extension-IA-Kuser-b}) {almost
surely} maintains its rank. First, we need to look at the core
elements $\mathrm{{\Phi}}_{23}^{-1}\mathrm{{\Phi}}_{ji}$ which
constitute $\mathrm{{V}}_{i}$. The objective is to prove all the
$\mathrm{{\Phi}}_{23}^{-1}\mathrm{{\Phi}}_{ji}$ terms are linear
independent. Then the expanded form is as in equation (\ref{Phi-2R})
and proved through the following three lemmas:

\normalsize
\begin{equation}
\begin{aligned}\label{Phi-2R}
&\mathrm{{\Phi}}_{23}^{-1}\mathrm{{\Phi}}_{ji}=\frac{(\mathrm{h}_{J_2R_1}\mathrm{G}_{1}\mathrm{h}_{R_1I_1}+\mathrm{h}_{J_2R_2}\mathrm{G}_{2}\mathrm{h}_{R_2I_1})}{(\mathrm{h}_{J_jR_1}\mathrm{G}_{1}\mathrm{h}_{R_1I_1}+\mathrm{h}_{J_jR_2}\mathrm{G}_{2}\mathrm{h}_{R_2I_1})}\bullet\\
&\frac{(\mathrm{h}_{J_jR_1}\mathrm{G}_{1}\mathrm{h}_{R_1I_i}+\mathrm{h}_{J_jR_2}\mathrm{G}_{2}\mathrm{h}_{R_2I_i})(\mathrm{h}_{J_1R_1}\mathrm{G}_{1}\mathrm{h}_{R_1I_3}+\mathrm{h}_{J_1R_2}\mathrm{G}_{2}\mathrm{h}_{R_2I_3})}{(\mathrm{h}_{J_2R_1}\mathrm{G}_{1}\mathrm{h}_{R_1I_3}+\mathrm{h}_{J_2R_2}\mathrm{G}_{2}\mathrm{h}_{R_2I_3})(\mathrm{h}_{J_1R_1}\mathrm{G}_{1}\mathrm{h}_{R_1I_i}+\mathrm{h}_{J_1R_2}\mathrm{G}_{2}\mathrm{h}_{R_2I_i})}\\
\end{aligned}
\end{equation}
\normalsize

\begin{Lemma}
$\mathrm{G}_1^3\mathbf{1}_T$,
$\mathrm{G}_1^2\mathrm{G}_2\mathbf{1}_T$,
$\mathrm{G}_1\mathrm{G}_2^2\mathbf{1}_T$,
$\mathrm{G}_2^3\mathbf{1}_T$ are linear independent almost surely.
\end{Lemma}
\begin{proof}
Let
$\mathrm{G}_1^{-1}\mathrm{G}_2=\mathrm{Diag}\{\lambda_1,\cdots,\lambda_T\}$.
$\mathrm{Diag}\{\cdot\}$ denotes the function to generate a diagonal
matrix with all elements in the set. Suppose the lemma is false,
then $(\mathrm{G}_1^{-1}\mathrm{G}_2)^3\mathbf{1}_T$,
$(\mathrm{G}_1^{-1}\mathrm{G}_2)^2\mathbf{1}_T$,
$\mathrm{G}_1^{-1}\mathrm{G}_2\mathbf{1}_T$, $\mathbf{1}_T$ are
linear dependent. So that there exist non-zero
$\alpha_1,\alpha_2,\alpha_3,\alpha_4$ satisfying the element-wise
equations of
$\alpha_1\lambda_t^3+\alpha_2\lambda_t^2+\alpha_3\lambda_t+\alpha_4=0$,
$\forall t=1,\ldots,T$. Since
$\mathrm{G}_1,\mathrm{G}_2,\lambda_1,\cdots,\lambda_T$ are random
generated independent parameters, and rank-3 polynomials could not
have $T$ non-identical roots almost surely, so it proves the lemma
is true almost surely.
\end{proof}

\begin{Lemma}\label{comb-Phi}
$\mathrm{{\Phi}}_{23}^{-1}\mathrm{{\Phi}}_{ji}\mathbf{1}_T$ and
$\mathrm{{\Phi}}_{23}^{-1}\mathrm{{\Phi}}_{nm}\mathbf{1}_T$ are
linear independent almost surely for arbitrary non-identical
$(j,i),(n,m)\in\mathcal{A}$, where $\mathcal{A}=\left\{(j,i)\ |\
j,i\in\mathcal{K}\setminus\{1\},j\neq i,(j,i)\neq(2,3)\right\}$.
\end{Lemma}

\begin{proof}

Suppose $\exists \text{ non-zero } \beta_a, \beta_b$ satisfying
$\beta_a\mathrm{{\Phi}}_{23}^{-1}\mathrm{{\Phi}}_{ji}\mathbf{1}_T+\beta_b\mathrm{{\Phi}}_{23}^{-1}\mathrm{{\Phi}}_{nm}\mathbf{1}_T=\mathbf{0}_{T}$,
where $\mathbf{0}_T\in\mathbb{C}^{T\times 1}$ denotes an all-zero
vector. Derived from (\ref{time-extension-IA-Kuser-a}), the
following equation is obtained:

\normalsize
\begin{equation}
\begin{aligned}\label{expand-1}
\beta_a\mathrm{\Theta}_{ji}\mathrm{\Theta}_{n1}\mathrm{\Theta}_{1m}\mathbf{1}_T+\beta_b\mathrm{\Theta}_{nm}\mathrm{\Theta}_{j1}\mathrm{\Theta}_{1i}\mathbf{1}_T=\mathbf{0}_{T}
\end{aligned}
\end{equation}
\normalsize

By using (\ref{Theta}) with (\ref{expand-1}), further we have
equation (\ref{expand-2}).

\begin{figure*}
\small
\begin{equation}
\begin{aligned}\label{expand-2}
&(\beta_a\mathrm{h}_{J_jR_1}\mathrm{h}_{R_1I_i}\mathrm{h}_{J_nR_1}\mathrm{h}_{R_1I_1}\mathrm{h}_{J_1R_1}\mathrm{h}_{R_1I_m}+\beta_b\mathrm{h}_{J_nR_1}\mathrm{h}_{R_1I_m}\mathrm{h}_{J_jR_1}\mathrm{h}_{R_1I_1}\mathrm{h}_{J_1R_1}\mathrm{h}_{R_1I_i})\cdot\mathrm{G}_1^3\mathbf{1}_T\\
&+[\beta_a(\mathrm{h}_{J_jR_1}\mathrm{h}_{R_1I_i}\mathrm{h}_{J_nR_1}\mathrm{h}_{R_1I_1}\mathrm{h}_{J_1R_2}\mathrm{h}_{R_2I_m}+\mathrm{h}_{J_jR_1}\mathrm{h}_{R_1I_i}\mathrm{h}_{J_nR_2}\mathrm{h}_{R_2I_1}\mathrm{h}_{J_1R_1}\mathrm{h}_{R_1I_m}+\\
&\hspace{9mm}\mathrm{h}_{J_jR_2}\mathrm{h}_{R_2I_i}\mathrm{h}_{J_nR_1}\mathrm{h}_{R_1I_1}\mathrm{h}_{J_1R_1}\mathrm{h}_{R_1I_m})+\beta_b(\mathrm{h}_{J_nR_1}\mathrm{h}_{R_1I_m}\mathrm{h}_{J_jR_1}\mathrm{h}_{R_1I_1}\mathrm{h}_{J_1R_2}\mathrm{h}_{R_2I_i}\\
&\hspace{5mm}+\mathrm{h}_{J_nR_1}\mathrm{h}_{R_1I_m}\mathrm{h}_{J_jR_2}\mathrm{h}_{R_2I_1}\mathrm{h}_{J_1R_1}\mathrm{h}_{R_1I_i}+\mathrm{h}_{J_nR_2}\mathrm{h}_{R_2I_m}\mathrm{h}_{J_jR_1}\mathrm{h}_{R_1I_1}\mathrm{h}_{J_1R_1}\mathrm{h}_{R_1I_i})]\cdot\mathrm{G}_1^2\mathrm{G}_2\mathbf{1}_T\\
&+[\beta_a(\mathrm{h}_{J_jR_1}\mathrm{h}_{R_1I_i}\mathrm{h}_{J_nR_2}\mathrm{h}_{R_2I_1}\mathrm{h}_{J_1R_2}\mathrm{h}_{R_2I_m}+\mathrm{h}_{J_jR_2}\mathrm{h}_{R_2I_i}\mathrm{h}_{J_nR_1}\mathrm{h}_{R_1I_1}\mathrm{h}_{J_1R_2}\mathrm{h}_{R_2I_m}+\\
&\hspace{9mm}\mathrm{h}_{J_jR_2}\mathrm{h}_{R_2I_i}\mathrm{h}_{J_nR_2}\mathrm{h}_{R_2I_1}\mathrm{h}_{J_1R_1}\mathrm{h}_{R_1I_m})+\beta_b(\mathrm{h}_{J_nR_1}\mathrm{h}_{R_1I_m}\mathrm{h}_{J_jR_2}\mathrm{h}_{R_2I_1}\mathrm{h}_{J_1R_2}\mathrm{h}_{R_2I_i}\\
&\hspace{5mm}+\mathrm{h}_{J_nR_2}\mathrm{h}_{R_2I_m}\mathrm{h}_{J_jR_1}\mathrm{h}_{R_1I_1}\mathrm{h}_{J_1R_2}\mathrm{h}_{R_2I_i}+\mathrm{h}_{J_nR_2}\mathrm{h}_{R_2I_m}\mathrm{h}_{J_jR_2}\mathrm{h}_{R_2I_1}\mathrm{h}_{J_1R_1}\mathrm{h}_{R_1I_i})]\cdot\mathrm{G}_1\mathrm{G}_2^2\mathbf{1}_T\\
&+(\beta_a\mathrm{h}_{J_jR_2}\mathrm{h}_{R_2I_i}\mathrm{h}_{J_nR_2}\mathrm{h}_{R_2I_1}\mathrm{h}_{J_1R_2}\mathrm{h}_{R_2I_m}+\beta_b\mathrm{h}_{J_nR_2}\mathrm{h}_{R_2I_m}\mathrm{h}_{J_jR_2}\mathrm{h}_{R_2I_1}\mathrm{h}_{J_1R_2}\mathrm{h}_{R_2I_i})\cdot\mathrm{G}_2^3\mathbf{1}_T\\
&=\mathbf{0}_{T}
\end{aligned}
\end{equation}
\normalsize \hrulefill \vspace{-0.5cm}
\end{figure*}


In (\ref{expand-2}), since $\mathrm{G}_1^3\mathbf{1}_T$,
$\mathrm{G}_1^2\mathrm{G}_2\mathbf{1}_T$,
$\mathrm{G}_1\mathrm{G}_2^2\mathbf{1}_T$,
$\mathrm{G}_2^3\mathbf{1}_T$ are actually four independent
$T$-dimensional bases, then all their coefficients must be zero. We
obtain another four linear equations about $\beta_a$, $\beta_b$,
with coefficients composed of $\mathrm{h}_{J_jR_1},
\mathrm{h}_{R_1I_i}, \mathrm{h}_{J_jR_2}, \mathrm{h}_{R_2I_i}$,
which are all random independent scalar values of the channels.
Therefore the solutions of $\beta_a,\beta_b$ are zero almost surely,
which contradicts the initial non-zero assumption. So that it proves
the lemma.


\end{proof}

\begin{Lemma}
All vectors in
$\{\mathrm{{\Phi}}_{23}^{-1}\mathrm{{\Phi}}_{ji}\mathbf{1}_T,
(j,i)\in\mathcal{A}\}$ are linear independent almost surely, where
$\mathcal{A}=\left\{(j,i)\ |\ j,i\in\mathcal{K}\setminus\{1\},j\neq
i,(j,i)\neq(2,3)\right\}$.\label{comb-Phi-multi}
\end{Lemma}

\begin{proof}

The procedure is similar to Lemma \ref{comb-Phi}. Only notice that
the set contains $N$ vectors. The new equation corresponding to
(\ref{expand-1}) has degree of $2N-1$ over $\Theta_{ji}$. Expand to
a new equation corresponding to (\ref{expand-2}), similarly there
are totally $2N$ $T$-dimensional bases such as
$\mathrm{G}_1^{2N}\mathbf{1}_T,
\mathrm{G}_1^{2N-1}\mathrm{G}_2^{1}\mathbf{1}_T$ etc. Then there are
$2N$ linear equations about $N$ variables $\beta_1,\ldots,\beta_N$,
which are forced to be zero. So that the independency is proved.
\end{proof}

\hspace{1mm}

Finally, in addition, consider the linear independency of the
exponentiation terms of
$\mathrm{{\Phi}}_{23}^{-1}\mathrm{{\Phi}}_{ji}\mathbf{1}_T$ in the
precoders $\mathrm{V}_i, \forall i\in\mathcal{K}$ in
(\ref{time-extension-IA-Kuser-b}). Because the signal dimension $T$
is large enough to afford the number of bases in space, so that all
$\mathrm{V}_i, \forall i\in\mathcal{K}$ are possibly
well-conditioned to implement interference alignment. However, it
needs further rigorous proof in our future work.

In summary, the comparison of the cases $L=1$ and $L\geq2$ in signal
vector space illustrates that in single antenna constant relay
channels, single relay is not able to provide channel
\textit{randomness} or \textit{relativity} which is the key to
interference alignment, while at least two relays are necessary to
provide this feature to approach 1/2 DoF (excluding duplex factor)
for every user.

\subsection{Number-Based Signal Level Alignment}

Besides signal vector alignment for single antenna constant
channels, there is another major class of schemes called signal
level interference alignment. There are several approaches to
investigate signal level alignment, including lattice coding,
deterministic models, and number theory. Among these approaches, IA
on the number domain is a very novel and canonical method. While our
focus in the following work is to show, in relay connected constant
channels, in a similar manner to vector alignment scheme, how level
alignment would also face the applicability issue of existing
designs when there is only one relay and the issue solved when there
are at least two relays.

\subsubsection{Basic Concept and Design Procedure}

At first, signal level could be viewed on the rational number scale
which represents infinite fractional DoF
\cite{real-interference-alignment-exploit-single-antenna}. Then
Khintchine-Groshev theorem also reveals that the field of real
numbers is rich enough to be equivalent to vector space to design
IA. Furthermore, \cite[Theorem 7]{DoF-compound-BC-finite-states}
uses a generalized version of Khintchin-Groshev theorem to extend
the designs to complex channels. In real channels, DoF is defined as
$d_{\mathrm{Real}}=\lim_{P\rightarrow
\infty}\frac{C_{\mathrm{sum}}}{1/2\log_2(P)}$ while in complex
channel DoF is defined as $d_{\mathrm{Complex}}=\lim_{P\rightarrow
\infty}\frac{C_{\mathrm{sum}}}{\log_2(P)}$ where $C_{\mathrm{sum}}$
is sum capacity and $P$ is transmission power.

By using this Motahari-Khandani scheme in
\cite{real-interference-alignment-exploit-single-antenna,DoF-compound-BC-finite-states},
$K$-user single antenna constant channels could approach $K/2$ DoF.
Since the objective in this work is to study the role of relays, the
core structure and procedure of IA are briefly described in a
complex channel setting.



Define $\mathbf{u}_k=[u_k^{(1)} u_k^{(2)} \cdots u_k^{(d_k)}]$ as
the number of $d_k$ datastreams sent from node $I_k$ to node $J_k$,
$k\in\mathcal{K}$. $u_k^{(m)}\in(-Q,Q)_{\mathbb Z},1\leq m\leq d_k$,
i.e. belongs to an integer constellation. Each datastream is
multiplied by a number $\nu_k^{(m)}\in\mathbb{C}$, which is called a
modulation pseudo-vector serving as distinct directions. In order to
satisfy the power constraint and control the minimum distance of the
received constellation, transmission signals should be scaled with a
constant $\lambda$.

In the meanwhile, each relay  node $R_l$ generates a random gain to
be a rational number $g_l\in\mathbb{Q}$ and the equivalent channel
from $I_i$ to $J_k$ is as following:

\normalsize
\begin{equation}\label{theta-number}
\begin{aligned}
\theta_{ki}=\sum_{l=1}^{L}\mathrm{h}_{J_kR_l}\cdot{g}_{l}\cdot\mathrm{h}_{R_lI_i}
\end{aligned}
\end{equation}
\normalsize




Then the received signal at destination node $J_k$ is:


\begin{equation}
\begin{aligned}\label{number-y-j}
y_k=\lambda\theta_{kk}\sum_{m=1}^{d_k}\nu_{k}^{(m)}u_{k}^{(m)}+\sum_{i=1,i\neq
k}^{K}\lambda\theta_{ki}\sum_{m=1}^{d_i}\nu_{i}^{(m)}u_{i}^{(m)}
\end{aligned}
\end{equation}

According to
\cite{real-interference-alignment-exploit-single-antenna,DoF-compound-BC-finite-states},
 the structure design of pseudo-vectors $\nu_k^{(m)}$ on number domain is similar
to the design of $\mathrm{V}_i$ in vector space of
\cite{beamforming-efficient-IA}, however there exists notable
difference as well. In this design, all $\nu_k^{(m)}, 1\leq m\leq
d_k$ belong to a set $\mathcal{B}_{\nu_k}$:

\normalsize
\begin{equation}\label{number-base-1}
\begin{aligned}
\mathcal{B}_{\nu_k}&=\left\{\prod_{j=1}^{K}\prod_{i=1,j\neq i}^{K}\theta_{ji}^{\alpha_{ji}}\Big|\begin{aligned}&0\leq\alpha_{ji}\leq n-1 \ \ \ i=k,j\neq k \\&0\leq\alpha_{ji}\leq n\ \ \ \ \ \ \ \ \text{Otherwise}\end{aligned}\right\}\\
\end{aligned}
\end{equation}
\normalsize

\vspace{0mm}

$n$ is integer, so that the number of streams
$d_k=|\mathcal{B}_{\nu_k}|=n^{K-1}(n+1)^{(K-1)^2}$.

For destination node $J_k$, the received signal space of
(\ref{number-y-j}) contains desired signal subspace formed by
$\theta_{kk}\mathcal{B}_{\nu_k}$ and interference subspace formed by
$\theta_{ki}\mathcal{B}_{\nu_i},i\neq k$. Observe all the
interference subspaces overlap to the same set of
$\mathcal{B}_{\nu_k}'$:

\normalsize
\begin{equation}\label{number-base-Intf}
\begin{aligned}
\mathcal{B}_{\nu_k}'&=\left\{\prod_{j=1}^{K}\prod_{i=1,j\neq i}^{K}\theta_{ji}^{\alpha_{ji}}\Big|0\leq\alpha_{ji}\leq n\right\}\\
\end{aligned}
\end{equation}
\normalsize

\vspace{0mm}

So that $|\mathcal{B}_{\nu_k}'|=(n+1)^{K(K-1)}$. Meanwhile the
desired signal subspace $\theta_{kk}\mathcal{B}_{\nu_k}$ and
interference subspace $\mathcal{B}_{\nu_k}'$ are distinct. According
to \cite[Theorem
6]{real-interference-alignment-exploit-single-antenna}, the total
DoF is
$\frac{Kn^{K-1}(n+1)^{(K-1)^2}}{n^{K-1}(n+1)^{(K-1)^2}+(n+1)^{K(K-1)}+1}$,
which approaches $K/2$ when $n$ is large.

\subsubsection{Issue in the Three-User Special Case}

In the three-user case, there is a special form of design in
\cite[Definition
1]{real-interference-alignment-exploit-single-antenna} presented as
following:

\normalsize
\begin{equation}\label{number-base-3user-standard}
\begin{aligned}
&\hspace{8mm}y_1=\vartheta_1\cdot x_1+x_2+x_3+z_1\\
&\hspace{8mm}y_2=\vartheta_2\cdot x_2+x_1+x_3+z_2\\
&\hspace{8mm}y_3=\vartheta_3\cdot x_3+x_1+\vartheta_0\cdot
x_2+z_3\\
&\vartheta_0=\frac{\theta_{13}\theta_{21}\theta_{32}}{\theta_{12}\theta_{23}\theta_{31}},
\vartheta_1=\frac{\theta_{11}\theta_{12}\theta_{23}}{\theta_{12}\theta_{21}\theta_{13}},
\vartheta_2=\frac{\theta_{22}\theta_{13}}{\theta_{12}\theta_{23}},
\vartheta_3=\frac{\theta_{33}\theta_{12}\theta_{21}}{\theta_{12}\theta_{23}\theta_{31}}
\end{aligned}
\end{equation}
\normalsize

So that user 1 has the transmit directions
$\{1,\vartheta_0,\vartheta_0^2,\ldots,\vartheta_0^n\}$, user 2 has
the directions of $\{1,\vartheta_0,\ldots,\vartheta_0^{n-1}\}$, and
user 3 has the directions of
$\{1,\vartheta_0,\ldots,\vartheta_0^{n-1}\}$.

$\bullet$\textit{ If the network has only one relay, $L=1$},
according to (\ref{theta-number}) and
(\ref{number-base-3user-standard}):

\normalsize
\begin{equation}\label{number-base-3user-vartheta-single}
\begin{aligned}
&\vartheta_0=\frac{\mathrm{h}_{J_1R_1}\mathrm{h}_{R_1I_3}\mathrm{h}_{J_2R_1}\mathrm{h}_{R_1I_1}\mathrm{h}_{J_3R_1}\mathrm{h}_{R_1I_2}\cdot{g}_{1}^3}{\mathrm{h}_{J_1R_1}\mathrm{h}_{R_1I_2}\mathrm{h}_{J_2R_1}\mathrm{h}_{R_1I_3}\mathrm{h}_{J_3R_1}\mathrm{h}_{R_1I_1}\cdot{g}_{1}^3}=1
\end{aligned}
\end{equation}
\normalsize

Then the encoding set
$\{1,\vartheta_0,\vartheta_0^2,\ldots,\vartheta_0^n\}$ used by all
three users collapses. The scheme of
\cite{real-interference-alignment-exploit-single-antenna,DoF-compound-BC-finite-states}
is not applicable to this case to support normal transmission with
interference alignment.



$\bullet$\textit{If the network has two relays, $L=2$}, IA scheme is
applicable as proved by the following lemma:

\begin{Lemma}\label{number-base-3user-L2}
$\vartheta_0\neq 1$ holds almost surely, for random independent
parameters $\mathrm{h}_{J_jR_1}$, $\mathrm{h}_{J_jR_2}$,
$\mathrm{h}_{R_1I_i}$, $\mathrm{h}_{R_2I_i}$, $g_1$, $g_2$, where
$i,j\in\mathcal{K}$.
\end{Lemma}
\begin{proof}
%

Suppose $\vartheta_0=1$, from (\ref{theta-number}) and
(\ref{number-base-3user-standard}), we have
${\theta_{13}\theta_{21}\theta_{32}}-{\theta_{12}\theta_{23}\theta_{31}}=0$,
and expand it to the equation
(\ref{number-base-3user-vartheta-double}).

\begin{figure*}
\small
\begin{equation}
\begin{aligned}\label{number-base-3user-vartheta-double}
&(\mathrm{h}_{J_1R_1}\mathrm{h}_{R_1I_3}\mathrm{h}_{J_2R_1}\mathrm{h}_{R_1I_1}\mathrm{h}_{J_3R_1}\mathrm{h}_{R_1I_2}-\mathrm{h}_{J_1R_1}\mathrm{h}_{R_1I_2}\mathrm{h}_{J_2R_1}\mathrm{h}_{R_1I_3}\mathrm{h}_{J_3R_1}\mathrm{h}_{R_1I_1})\cdot g_1^3\\
&+[(\mathrm{h}_{J_1R_1}\mathrm{h}_{R_1I_3}\mathrm{h}_{J_2R_1}\mathrm{h}_{R_1I_1}\mathrm{h}_{J_3R_2}\mathrm{h}_{R_2I_2}+\mathrm{h}_{J_1R_1}\mathrm{h}_{R_1I_3}\mathrm{h}_{J_2R_2}\mathrm{h}_{R_2I_1}\mathrm{h}_{J_3R_1}\mathrm{h}_{R_1I_2}+\\
&\hspace{9mm}\mathrm{h}_{J_1R_2}\mathrm{h}_{R_2I_3}\mathrm{h}_{J_2R_1}\mathrm{h}_{R_1I_1}\mathrm{h}_{J_3R_1}\mathrm{h}_{R_1I_2})-(\mathrm{h}_{J_1R_1}\mathrm{h}_{R_1I_2}\mathrm{h}_{J_2R_1}\mathrm{h}_{R_1I_3}\mathrm{h}_{J_3R_2}\mathrm{h}_{R_2I_1}\\
&\hspace{5mm}+\mathrm{h}_{J_1R_1}\mathrm{h}_{R_1I_2}\mathrm{h}_{J_2R_2}\mathrm{h}_{R_2I_3}\mathrm{h}_{J_3R_1}\mathrm{h}_{R_1I_1}+\mathrm{h}_{J_1R_2}\mathrm{h}_{R_2I_2}\mathrm{h}_{J_2R_1}\mathrm{h}_{R_1I_3}\mathrm{h}_{J_3R_1}\mathrm{h}_{R_1I_1})]\cdot g_1^2g_2\\
&+[(\mathrm{h}_{J_1R_1}\mathrm{h}_{R_1I_3}\mathrm{h}_{J_2R_2}\mathrm{h}_{R_2I_1}\mathrm{h}_{J_3R_2}\mathrm{h}_{R_2I_2}+\mathrm{h}_{J_1R_2}\mathrm{h}_{R_2I_3}\mathrm{h}_{J_2R_1}\mathrm{h}_{R_1I_1}\mathrm{h}_{J_3R_2}\mathrm{h}_{R_2I_2}+\\
&\hspace{9mm}\mathrm{h}_{J_1R_2}\mathrm{h}_{R_2I_3}\mathrm{h}_{J_2R_2}\mathrm{h}_{R_2I_1}\mathrm{h}_{J_3R_1}\mathrm{h}_{R_1I_2})-(\mathrm{h}_{J_1R_1}\mathrm{h}_{R_1I_2}\mathrm{h}_{J_2R_2}\mathrm{h}_{R_2I_3}\mathrm{h}_{J_3R_2}\mathrm{h}_{R_2I_1}\\
&\hspace{5mm}+\mathrm{h}_{J_1R_2}\mathrm{h}_{R_2I_2}\mathrm{h}_{J_2R_1}\mathrm{h}_{R_1I_3}\mathrm{h}_{J_3R_2}\mathrm{h}_{R_2I_1}+\mathrm{h}_{J_1R_2}\mathrm{h}_{R_2I_2}\mathrm{h}_{J_2R_2}\mathrm{h}_{R_2I_3}\mathrm{h}_{J_3R_1}\mathrm{h}_{R_1I_1})]\cdot g_1g_2^2\\
&+(\mathrm{h}_{J_1R_2}\mathrm{h}_{R_2I_3}\mathrm{h}_{J_2R_2}\mathrm{h}_{R_2I_1}\mathrm{h}_{J_3R_2}\mathrm{h}_{R_2I_2}-\mathrm{h}_{J_1R_2}\mathrm{h}_{R_2I_2}\mathrm{h}_{J_2R_2}\mathrm{h}_{R_2I_3}\mathrm{h}_{J_3R_2}\mathrm{h}_{R_2I_1})\cdot g_2^3\\
&=0
\end{aligned}
\end{equation}
\normalsize \hrulefill \vspace{-0.5cm}
\end{figure*}

In equation (\ref{number-base-3user-vartheta-double}), since
$\{g_1^3 , g_1^2g_2 , g_1g_2^2 , g_2^3\}$ are independent number
bases with measure of one, the corresponding coefficients must be
all zeros. Observe $g_1^3$ and $g_2^3$ in
(\ref{number-base-3user-vartheta-double}), their coefficients are
indeed 0; while observe $g_1^2g_2$ and $g_1g_2^2$, their
coefficients are non-zero with measure one, because of random
selected parameters of $\mathrm{h}_{J_jR_1}$, $\mathrm{h}_{J_jR_2}$,
$\mathrm{h}_{R_1I_i}$, $\mathrm{h}_{R_2I_i}$. As a result, it proves
that $\vartheta_0\neq 1$.


\end{proof}

In summary, the comparison of the cases $L=1$ and $L=2$ shows that
single relay is not able to provide channel \textit{randomness} or
\textit{relativity} in the signal level alignment on number domain,
while two relays are capable of provide this feature to approach 1/2
due DoF for every user.

\subsubsection{Issue in the $K$-User General Case}
Extend the study of the three-user case to the general $K$-user
case.

$\bullet$\textit{If the network has only one relay, $L=1$},
according to (\ref{theta-number}) and (\ref{number-base-1}), the set
$\mathcal{B}_{\nu_k}$ is presented in an equal form of set
$\mathcal{B}_{\nu_k}^{L1}$:

\vspace{0mm}

\normalsize
\begin{equation}\label{number-base-1-single}
\begin{aligned}
\mathcal{B}_{\nu_k}^{L1}&=\left\{\prod_{j=1}^{K}\prod_{i=1,j\neq i}^{K}(\mathrm{h}_{J_jR_1}\cdot{g}_{1}\cdot\mathrm{h}_{R_1I_i})^{\alpha_{ji}}\Big|\begin{aligned}&0\leq\alpha_{ji}\leq n-1 \ i=k,j\neq k \\&0\leq\alpha_{ji}\leq n\ \ \ \ \ \ \text{Otherwise}\end{aligned}\right\}\\
\end{aligned}
\end{equation}
\normalsize

Counting the cardinality $|\mathcal{B}_{\nu_k}^{L1}|$ is complicated
so that it is natural to think of attain an upperbound by relaxing
and expanding the set $\mathcal{B}_{\nu_k}^{L1}$ to a new set
$\bar{\mathcal{B}}_{\nu_k}^{L1}$.
$\mathcal{B}_{\nu_k}^{L1}\subset\bar{\mathcal{B}}_{\nu_k}^{L1}$:

\normalsize
\begin{equation}\label{number-base-1-single-upper}
\begin{aligned}
\bar{\mathcal{B}}_{\nu_k}^{L1}&=\left\{\big(\prod_{j\in\mathcal{K}}\mathrm{h}_{J_jR_1}^{\tau_{j}}\big)\cdot\big({g}_{1}^{\gamma}\big)\cdot\big(\prod_{i\in\mathcal{K}}\mathrm{h}_{R_1I_i}^{\varsigma_i}\big)\Big|\begin{aligned}&0\leq\gamma\leq nK^2,0\leq\tau_j,\varsigma_i\leq nK\\&\sum_{j\in\mathcal{K}}\tau_j= \gamma, {\sum_{i\in\mathcal{K}}\varsigma_i}= \gamma\end{aligned}\right\}\\
\end{aligned}
\end{equation}
\normalsize

Notice that
$\mathcal{B}_{\nu_k}^{L1}\subset\bar{\mathcal{B}}_{\nu_k}^{L1}$, so
that $|\mathcal{B}_{\nu_k}^{L1}|<|\bar{\mathcal{B}}_{\nu_k}^{L1}|$.
Further loose the set by:

\vspace{0mm}

\normalsize
\begin{equation}
\begin{aligned}\label{number-loosebound}
&\mathcal{C}_{\gamma}=\{(\tau_1,\tau_2,\ldots,\tau_K)|0\leq\tau_j\leq
nK,\forall j\in\mathcal{K},\sum_{j\in\mathcal{K}}\tau_j=\gamma\}\\
&\mathcal{D}_{\gamma}=\{(\tau_1,\tau_2,\ldots,\varsigma_K)|0\leq\varsigma_i\leq
nK,\forall i\in\mathcal{K},\sum_{i\in\mathcal{K}}\varsigma_i=\gamma\}\\
&|\bar{\mathcal{B}}_{\nu_k}^{L1}|\leq\sum_{\gamma=0}^{nK^2}{|\mathcal{C}_{\gamma}|\cdot|\mathcal{D}_{\gamma}|}<\sum_{\gamma=0}^{nK^2}(nK+1)^K\cdot(nK+1)^K\\
\end{aligned}
\end{equation}
\normalsize

\vspace{0mm}

Compare with the original due cardinality of (\ref{number-base-1}):

%
\normalsize
\begin{equation}
\begin{aligned}\label{number-card-compare}
\frac{|\mathcal{B}_{\nu_k}^{L1}|}{n^{K-1}(n+1)^{(K-1)^2}}<\frac{(nK^2+1)(nK+1)^{2K}}{n^{K-1}(n+1)^{(K-1)^2}}\xrightarrow[]{n\rightarrow
\infty, K\rightarrow \infty}0
\end{aligned}
\end{equation}
\normalsize

It reveals that in equivalent channels generated by single relay,
$L=1$, the encoding set $\mathcal{B}_{\nu_k}$ degenerates to
negligible size for large $n$ and $K$, i.e. Motahari-Khandani scheme
\cite{real-interference-alignment-exploit-single-antenna,DoF-compound-BC-finite-states}
is not applicable in this case to support normal transmission with
interference alignment.

$\bullet$\textit{If the network has at least two relays, $L\geq2$},
we have similar conclusions as lemma \ref{number-base-3user-L2} such
that the precoding set are greatly expanded with independent
elements to support normal transmission with interference alignment.
It also needs further comprehensive validation and rigorous proof
with exact requisite conditions.

\subsection{Asymmetric Complex Signaling Alignment}

Asymmetric complex signaling alignment is categorized into a class
of phase alignment. It is proposed by Jafar as an effective measure
to solve the problem of obtaining high DoF in constant interference
channels
\cite{IA-asymmetric-complex-signaling-settling-nosratinia-conjecture-trans}.
However, the work in this thesis investigates if it is applicable to
the relay generated equivalent interference channel.

It has been conjectured by H$\o$st, Madsen and Nosratinia that
complex Gaussian interference channels with single antenna nodes and
constant channel coefficients have only one degree-of-freedom
regardless of the number of users \cite{multiplexing-gain-wireless}.
Intuitively, this conjecture indicates the optimality of orthogonal
medium access such as TDMA where each user is assigned a fraction of
DoF and the sum DoF is equal to one. Original interference alignment
\cite{IA-DOF-Kuser-Interference,IA-Spacial-DOF-Kuser-Interference}
provided a powerful tool to achieve $K/2$ DoF for $K$-user
interference channels, only under the condition of time-varying or
frequency-selective channel coefficients. However, in single antenna
constant channels, it was not known if H$\o$st-Madsen-Nosratinia
conjecture is right until the appearance of asymmetric complex
signaling method.

By using asymmetric complex signal inputs, Jafar shows that at least
1.2 DoF are achievable on the complex 3-user interference channel
with constant channel coefficients
\cite{IA-asymmetric-complex-signaling-settling-nosratinia-conjecture-trans},
without any assumption of time-variations/frequency-selectivity used
in prior work. In conventional wireless systems only circularly
symmetric complex Gaussian random variables are typically used in
order to maximize entropy, while in asymmetric complex signaling the
inputs are chosen to be complex but not circularly symmetric. For
example, in an specific case of phase alignment, each transmitter
uses only real valued Gaussian inputs, and ensures that interference
at each receiver aligns in the imaginary dimension when the desired
signal is in the real dimension of the complex space. This idea
could be extended to general values of channel coefficients
thereafter.

Detailed introduction is omitted here scheme, while the basic model
and conclusion is briefly presented. First, each received complex
signal is expressed in real values as following (noise is ignored
for simplicity):

\begin{equation}\label{asymmetri-complex-equation}
\begin{aligned}
\left[\begin{aligned}\mathrm{Re}\{\mathrm{Y}_i\}\\\mathrm{Im}\{\mathrm{Y}_i\}\end{aligned}\right]=\sum_{j=1}^{K}|\mathrm{H}_{ij}|\underbrace{\left[\begin{aligned}&\cos(\phi_{ij})\
-&\sin(\phi_{ij})\\&\sin(\phi_{ij})\
&\cos(\phi_{ij})\end{aligned}\right]}_{\text{Defined as
}U(\phi_{ij})}\left[\begin{aligned}\mathrm{Re}\{\mathrm{X}_j\}\\\mathrm{Im}\{\mathrm{X}_j\}\end{aligned}\right]
\end{aligned}
\end{equation}

where the complex output $\mathrm{Y}_i$ and complex input
$\mathrm{X}_j$ are split into scalar counterparts of real values,
and the complex scalar channel $\mathrm{H}_{ij}$ is expressed by its
amplitude $|\mathrm{H}_{ij}|$ and a transformation matrix $U(\phi)$,
which is a rotation matrix defined with the following properties:

\begin{equation}\label{asymmetri-complex-phi}
\begin{aligned}
U(\phi)^{-1}&=U(-\phi)\\
U(\phi)U(\theta)&=U(\theta)U(\phi)\\
U(\phi)U(\theta)&=U(\phi+\theta)
\end{aligned}
\end{equation}

With the new formulation, it is proved that the 3 user complex
Gaussian interference channel with constant coefficients achieves
1.2 degrees-of-freedom if all of the following conditions are
satisfied \cite[Theorem
2]{IA-asymmetric-complex-signaling-settling-nosratinia-conjecture-trans}:

\begin{equation}\label{asymmetri-complex-condition}
\begin{aligned}
\phi_{J_1I_3}+\phi_{J_2I_1}-\phi_{J_2I_3}-\phi_{J_1I_1}\ \neq 0\ \mod{(\pi)}\\
\phi_{J_1I_2}+\phi_{J_3I_1}-\phi_{J_3I_2}-\phi_{J_1I_1}\ \neq 0\ \mod{(\pi)}\\
\phi_{J_2I_1}+\phi_{J_3I_2}-\phi_{J_3I_1}-\phi_{J_2I_2}\ \neq 0\ \mod{(\pi)}\\
\phi_{J_2I_3}+\phi_{J_1I_2}-\phi_{J_1I_3}-\phi_{J_2I_2}\ \neq 0\ \mod{(\pi)}\\
\phi_{J_3I_2}+\phi_{J_1I_3}-\phi_{J_1I_2}-\phi_{J_3I_3}\ \neq 0\ \mod{(\pi)}\\
\phi_{J_3I_1}+\phi_{J_2I_3}-\phi_{J_2I_1}-\phi_{J_3I_3}\ \neq 0\ \mod{(\pi)}\\
\end{aligned}
\end{equation}

where $I_1$, $I_2$, $I_3$ are transmitter nodes and $J_1$, $J_2$,
$J_3$ are receiver nodes. Notice the coditions of
(\ref{asymmetri-complex-condition}) are naturally satisfied almost
surely for general 3-user networks. However, consider the relay
networks as in Fig. \ref{relay-IA-pic2}, if there is only one relay,
then the equivalent channel phase from transmitter $I_i$ to receiver
$J_j$ is equal to the sum of phases from transmitter $I_i$ to relay
$R_1$ and from relay $R_1$ to receiver $J_j$:

\begin{equation}\label{asymmetri-complex-phi-relay}
\begin{aligned}
\phi_{J_jI_i}=\phi_{R_1I_i}+\phi_{J_jR_1}\ \mod{(\pi)}\\
\end{aligned}
\end{equation}

Therefore, check the conditions of
(\ref{asymmetri-complex-condition}) again with
(\ref{asymmetri-complex-phi-relay}), and it is easy to find out the
following new results which make the design no longer available:

\begin{equation}\label{asymmetri-complex-condition-relay}
\begin{aligned}
\phi_{J_1I_3}+\phi_{J_2I_1}-\phi_{J_2I_3}-\phi_{J_1I_1}\ = 0\ \mod{(\pi)}\\
\phi_{J_1I_2}+\phi_{J_3I_1}-\phi_{J_3I_2}-\phi_{J_1I_1}\ = 0\ \mod{(\pi)}\\
\phi_{J_2I_1}+\phi_{J_3I_2}-\phi_{J_3I_1}-\phi_{J_2I_2}\ = 0\ \mod{(\pi)}\\
\phi_{J_2I_3}+\phi_{J_1I_2}-\phi_{J_1I_3}-\phi_{J_2I_2}\ = 0\ \mod{(\pi)}\\
\phi_{J_3I_2}+\phi_{J_1I_3}-\phi_{J_1I_2}-\phi_{J_3I_3}\ = 0\ \mod{(\pi)}\\
\phi_{J_3I_1}+\phi_{J_2I_3}-\phi_{J_2I_1}-\phi_{J_3I_3}\ = 0\ \mod{(\pi)}\\
\end{aligned}
\end{equation}

Similar to the time-extended signal vector scheme and the
number-based signal level scheme, the remark for complex phase
alignment is concluded:

\begin{Remark}
For single antenna constant channels, if the $K$-pair users are
connected by only one relay, then asymmetric complex signaling
scheme would not work to obtain more than one DoF in this relay
generated equivalent channel.
\end{Remark}

\section{Solvement of Single Relay Issue in Precoding Scheme in Constant Channels}

Previous analysis shows that single relay is not capable of
providing channel randomness to support interference alignment
schemes by precoding at transmitters. However, in this section, we
look into possible alternative solutions in the scenarios only with
single relay. In order to address this issue, a natural strategy is
to transform the role of multiple ($L\geq 2$) relays to the role of
multiple ($L\geq 2$) time slots scenario with one relay.

As a primitive investigation, at first the following part only looks
into the case when there are only two separate network statuses. In
particular, the case is equivalent to either having two relays, i.e.
$L=2$, or dividing the whole transmission period into $L=2$ parts.
Suppose the transmission period starts from time slot 1 to time slot
$T$ (set $T$ as even). To divide all the $T$ time slots into $L=2$
parts, just think of the channel statuses at odd time slots as one
virtual network with one virtual relay; then think of the channel
statuses at even time slots as another virtual network with another
virtual relay; add these two classes of statuses correspondingly to
form a virtual network with two relays. Then the equivalent channel
as shown in previous equation (\ref{Theta-diag}) is presented as:

\normalsize
\begin{equation}\label{Theta-diag-doublelayer}
\begin{aligned}
\mathrm{\Theta}_{ji}=\ \ \
&\mathrm{Diag}\{\mathrm{f}_{jl}^{[1]}{\mathrm{g}_{l}}(1,1)\mathrm{h}_{li}^{[1]},
\mathrm{f}_{jl}^{[3]}\mathrm{g}_{l}(3,3)\mathrm{h}_{li}^{[3]},
\cdots,
\mathrm{f}_{jl}^{[T-1]}\mathrm{g}_{l}(T-1,T-1)\mathrm{h}_{li}^{[T-1]}\}\\
+&\mathrm{Diag}\{\mathrm{f}_{jl}^{[2]}{\mathrm{g}_{l}}(2,2)\mathrm{h}_{li}^{[2]},
\mathrm{f}_{jl}^{[4]}\mathrm{g}_{l}(4,4)\mathrm{h}_{li}^{[4]},
\cdots,
\mathrm{f}_{jl}^{[T]}\mathrm{g}_{l}(T,T)\mathrm{h}_{li}^{[T]}\}
\end{aligned}
\end{equation}
\normalsize

In (\ref{Theta-diag-doublelayer}), actually $l=1$ because there is
actually one relay. $\mathrm{h}_{li}^{[t]}$ represents the channel
from source node $i$ to relay node $l$ at the time slot $t$, and
$\mathrm{f}_{jl}^{[t]}$ represents the channel from relay node $l$
to destination node $j$ at the time slot $t$.

Notice that an issue arises here. Because the channels are constant
in all $T$ time slots, so that
$\mathrm{h}_{li}^{[1]}=\mathrm{h}_{li}^{[2]}=\cdots=\mathrm{h}_{li}^{[T]}$
and
$\mathrm{f}_{jl}^{[1]}=\mathrm{f}_{jl}^{[2]}=\cdots=\mathrm{f}_{jl}^{[T]}$.
Then the above equivalent channel (\ref{Theta-diag-doublelayer})
becomes:

\normalsize
\begin{equation}\label{Theta-diag-doublelayer-constant}
\begin{aligned}
\mathrm{\Theta}_{ji}=\ \ \
&\mathrm{f}_{jl}^{[1]}\mathrm{h}_{li}^{[1]}\mathrm{Diag}\{{\mathrm{g}_{l}}(1,1),
\mathrm{g}_{l}(3,3), \cdots,
\mathrm{g}_{l}(T-1,T-1)\}\\
+&\mathrm{f}_{jl}^{[1]}\mathrm{h}_{li}^{[1]}\mathrm{Diag}\{{\mathrm{g}_{l}}(2,2),
\mathrm{g}_{l}(4,4), \cdots, \mathrm{g}_{l}(T,T)\}\\
\triangleq\ \ \
&\mathrm{f}_{jl}^{[1]}\mathrm{h}_{li}^{[1]}\mathrm{Diag}\{g(1),
g(2), \cdots, g(T/2)\}
\end{aligned}
\end{equation}
\normalsize

Observe the final combined result in
(\ref{Theta-diag-doublelayer-constant}) and find that the channel is
equivalent to previous case of (\ref{Theta-SR}) which has only
single relay. In this case, it has been proved interference
alignment is definitely not applicable at all.

Therefore, to artificially emulate the case of two relays where the
channels should show randomness, two solutions are proposed in the
following two parts. One is by using switchable antenna for the
relay node in the middle, so that each time after the antenna is
switched all the corresponding channels change together in a random
manner; the other is to proceed fluctuation precoding scheme for
both end nodes at the edge with double-layered symbol extensions, so
that each time the precoding gains change all the channels change
randomly as well.

\subsection{Antenna-Switching Coding in the Middle}


The solution of antenna-switching coding is inspired by the idea of
reconfigurable antenna and \textit{blind} interference alignment.
Reconfigurable antenna is an emerging technology, changing its
characteristics by switching geometrical-metallic segments to create
pre-determined independent modes in every time slot
\cite{nano-electromech-switch-reconfig-antenna}. Each distinct
configuration is a different mode. This technology is recently
introduced to the research of \textit{blind} interference alignment
\cite{aiming-dark-blind-IA-staggered-antenna}, which utilizes
reconfigurable antenna to manipulate the channel directly to create
channel fluctuation patterns that are exploited by the transmitter.

Therefore, the relay $R_l$ ($l=1$) is set as a reconfigurable
antenna and it has two modes which operate for odd time slots and
even time slots respectively. An important feature is that channels
change between odd time slots and even time slots, i.e.
$\mathrm{h}_{li}^{[1]}\neq \mathrm{h}_{li}^{[2]}$ and
$\mathrm{f}_{jl}^{[1]}\neq \mathrm{f}_{jl}^{[2]}$ while
$\mathrm{h}_{li}^{[1]}=\mathrm{h}_{li}^{[3]}=\cdots=\mathrm{h}_{li}^{[T-1]}$,
$\mathrm{h}_{li}^{[2]}=\mathrm{h}_{li}^{[4]}=\cdots=\mathrm{h}_{li}^{[T]}$
and
$\mathrm{f}_{jl}^{[1]}=\mathrm{f}_{jl}^{[3]}=\cdots=\mathrm{f}_{jl}^{[T-1]}$
and
$\mathrm{f}_{jl}^{[2]}=\mathrm{f}_{jl}^{[4]}=\cdots=\mathrm{f}_{jl}^{[T]}$.
Then the above equivalent channel (\ref{Theta-diag-doublelayer})
becomes:

\normalsize
\begin{equation}\label{Theta-diag-doublelayer-blind}
\begin{aligned}
\mathrm{\Theta}_{ji}=\ \ \
&\mathrm{f}_{jl}^{[1]}\mathrm{h}_{li}^{[1]}\mathrm{Diag}\{{\mathrm{g}_{l}}(1,1),
\mathrm{g}_{l}(3,3), \cdots,
\mathrm{g}_{l}(T-1,T-1)\}\\
+&\mathrm{f}_{jl}^{[2]}\mathrm{h}_{li}^{[2]}\mathrm{Diag}\{{\mathrm{g}_{l}}(2,2),
\mathrm{g}_{l}(4,4), \cdots, \mathrm{g}_{l}(T,T)\}
\end{aligned}
\end{equation}
\normalsize

%

Since it could not be further combined, so that it is equivalent to
the case with two relays in equation (\ref{Theta-DR}). Interference
alignment is much likely to be applied in this case. However,
regarding DoF, there is an additional factor $1/2$ due to the
double-layered symbol extension. Thus the $K$-pair network could
approach $K/4$ (excluding duplex factor) DoF with one single antenna
relay on condition that IA is finally successful. Notice the DoF is
still scalable regarding $K$.

\subsection{Antenna-Fluctuating Coding at the Edge}


The above antenna-switching solution has an extra requirement of the
antenna being physically reconfigured. While in this part, a simple
but effective solution is proposed to directly proceed coding at the
end nodes, instead of the relay. Let each source node $i$ has an
additional random gain of $\alpha_i^{[t]}$ at the time slot $t$; and
each destination node $j$ has an additional random gain of
$\beta_j^{[t]}$ at the time slot $t$. Then the equivalent channel
(\ref{Theta-diag-doublelayer}) becomes:

\normalsize
\begin{equation}\label{Theta-diag-doublelayer-edge}
\begin{aligned}
\mathrm{\Theta}_{ji}=\ \ \
&\mathrm{Diag}\{\beta_j^{[1]}\mathrm{f}_{jl}^{[1]}{\mathrm{g}_{l}}(1,1)\mathrm{h}_{li}^{[1]}\alpha_i^{[1]},
\beta_j^{[3]}\mathrm{f}_{jl}^{[3]}\mathrm{g}_{l}(3,3)\mathrm{h}_{li}^{[3]}\alpha_i^{[3]},
\\
&\hspace{60mm}\cdots,
\beta_j^{[T-1]}\mathrm{f}_{jl}^{[T-1]}\mathrm{g}_{l}(T-1,T-1)\mathrm{h}_{li}^{[T-1]}\alpha_i^{[T-1]}\}\\
+&\mathrm{Diag}\{\beta_j^{[2]}\mathrm{f}_{jl}^{[2]}{\mathrm{g}_{l}}(2,2)\mathrm{h}_{li}^{[2]}\alpha_i^{[2]},
\beta_j^{[4]}\mathrm{f}_{jl}^{[4]}\mathrm{g}_{l}(4,4)\mathrm{h}_{li}^{[4]}\alpha_i^{[4]},
\\
&\hspace{75mm}\cdots,
\beta_j^{[T]}\mathrm{f}_{jl}^{[T]}\mathrm{g}_{l}(T,T)\mathrm{h}_{li}^{[T]}\alpha_i^{[T]}\}
\end{aligned}
\end{equation}
\normalsize

Similarly, let $\alpha_i^{[1]}\neq \alpha_i^{[2]}$ and
$\beta_j^{[1]}\neq \beta_j^{[2]}$, while
$\alpha_i^{[1]}=\alpha_i^{[3]}=\cdots\alpha_i^{[T-1]}$,
$\alpha_i^{[2]}=\alpha_i^{[4]}=\cdots\alpha_i^{[T]}$ and
$\beta_j^{[1]}=\beta_j^{[3]}=\cdots=\beta_j^{[T-1]}$,
$\beta_j^{[2]}=\beta_j^{[4]}=\cdots=\beta_j^{[T]}$. Besides, because
channels are constant,
$\mathrm{h}_{li}^{[1]}=\mathrm{h}_{li}^{[2]}=\cdots=\mathrm{h}_{li}^{[T]}$
and
$\mathrm{f}_{jl}^{[1]}=\mathrm{f}_{jl}^{[2]}=\cdots=\mathrm{f}_{jl}^{[T]}$.
So that (\ref{Theta-diag-doublelayer-edge}) becomes:

\normalsize
\begin{equation}\label{Theta-diag-doublelayer-edge-2}
\begin{aligned}
\mathrm{\Theta}_{ji}=\ \ \
&\beta_j^{[1]}\mathrm{f}_{jl}^{[1]}\mathrm{h}_{li}^{[1]}\alpha_i^{[1]}\mathrm{Diag}\{{\mathrm{g}_{l}}(1,1),
\mathrm{g}_{l}(3,3),\cdots,
\mathrm{g}_{l}(T-1,T-1)\}\\
+&\beta_j^{[2]}\mathrm{f}_{jl}^{[1]}\mathrm{h}_{li}^{[1]}\alpha_i^{[2]}\mathrm{Diag}\{{\mathrm{g}_{l}}(2,2),
\mathrm{g}_{l}(4,4),\cdots, \mathrm{g}_{l}(T,T)\}
\end{aligned}
\end{equation}
\normalsize

Then it is also equivalent to the case with two relays in equation
(\ref{Theta-DR}). Interference alignment is much likely to be
applied in this case. Also, the $K$-pair network could only approach
$K/4$ (excluding duplex factor) DoF which is scalable regarding $K$
on condition IA is finally successful. Notice that in this solution
reconfigurable antenna is not required. Simple coding at the edge
could generate channel randomness to apply IA schemes.

\section{Numerical Relay Coding Scheme aiming at DoF}

Although the analytical design could approach $K/2$ DoF, it is very
complicated to implement and requires very large dimensions in
signal subspaces so that it is not quite practical. Thus, in the
following part, a numerical method is proposed for the relay coding,
corresponding to the strategy of \textit{coding in the middle}
\cite{NC-Three-Unicast-IA}.

Relay coding and optimization methods have been studied in a lot of
works. However, all the methods take rates as direct objectives.
While concerning the high SNR region, in term of DoF, a novel idea
is naturally considered to optimize the achievable DoF directly with
relay coding. However, it is necessary to find an appropriate tool
to support the design. Therefore, in this work, it is proposed to
utilize a powerful tool called rank constrained rank minimization
(RCRM) \cite{IA-RCRM,IA-RCRM-Globcom}. While RCRM method is
originally used to design the precoding scheme to achieve
interference alignment in general $K$-user interference channels. So
the previous solution is not applicable to the relay coding problem
in our scenario.

In this work, the contributions are two-fold: first, the RCRM
formulation \cite{IA-RCRM} is transformed to a novel problem
(\ref{convx-approx-1and2}) regarding new decision variables of
$\mathrm{G}_l$, and it is proved that its convexity is maintained to
proceed numerical algorithms; second, along with a variety of
extensions of this scenario, a general relay coding algorithm is
presented to solve such kind of problems in a universal practicable
way.

\subsection{Interference Alignment by Rank Constrained Rank Minimization (RCRM)}

Consider the model in Fig. \ref{relay-IA} with all nodes equipped
with single antenna. $K$ sources communicate with $K$ destinations
via $L$ relays. Time-extended MIMO scheme is used with a period time
of $T$. The relay coding matrix is $\mathrm{G}_l$ as in (\ref{G}).
Then the equivalent channel $\Theta_{ji}$ is shown in (\ref{Theta}),
composed of $\mathrm{G}_l$. Each user has $d$ streams in the scheme,
so that $\mathrm{V}_{i}\in\mathbb{C}^{T\times d}$ is the precoding
matrix at source $i$; $\mathrm{G}_{l}\in\mathbb{C}^{T\times T}$;
$\mathrm{U}_{j}\in\mathbb{C}^{T\times d}$ is the receive filter
matrix at destination $j$. Then each receiver $j$ linearly processes
the received signal by zero-forcing. The received interference and
signal should satisfy the following conditions:

\begin{equation}\label{DoF-cond-1}
\begin{aligned}
\mathrm{U}_j^{\dag}\mathrm{\Theta}_{ji}\mathrm{V}_i&=\mathbf{0}_{d\times
d} \qquad i\neq j \\
\mathrm{rank}(\mathrm{U}_j^{\dag}\mathrm{\Theta}_{jj}\mathrm{V}_j)&=d
\end{aligned}
\end{equation}

Define interference and signal subspace respectively for receiver
$j$:

\normalsize
\begin{equation}\label{Def-JS}
\begin{aligned}
\mathcal{J}_j(\{\mathrm{V}_i\}_{i\neq j}^{K},\{\mathrm{G}_{l}\}_{l=1}^L,\mathrm{U}_j)&\triangleq \mathrm{U}_j^{\dag}[\{\mathrm{\Theta}_{ji}\mathrm{V}_i\}_{i\neq j}^{K}]\\
\mathcal{S}_j(\mathrm{V}_j,\{\mathrm{G}_{l}\}_{l=1}^L,\mathrm{U}_j)&\triangleq\mathrm{U}_j^{\dag}\mathrm{\Theta}_{jj}\mathrm{V}_j
\end{aligned}
\end{equation}
\normalsize

where $\mathrm{J}_j(\{\mathrm{V}_i\}_{i\neq
j}^{K},\{\mathrm{G}_{l}\}_{l=1}^L,\mathrm{U}_j)$ gathers all
interference observed at the $j$-th destination and
$\mathrm{S}_j(\mathrm{V}_j,\{\mathrm{G}_{l}\}_{l=1}^L,\mathrm{U}_j)$
gathers the desired signal from the $j$-th source. However, compared
with original problem in \cite{IA-RCRM}, notice in our work, the
formulation (\ref{Def-JS}) contains $\mathrm{G}_{l}$ as new
variables through the equivalent channel $\Theta_{ji}$ as in
(\ref{Theta}), in addition to prior variables $\mathrm{V}_i$ and
$\mathrm{U}_j$.

The conventional method of IA to approach zero-forcing conditions in
(\ref{DoF-cond-1}) is the interference leakage minimization
algorithm \cite{approaching-capacity-IA}. While the idea of
\cite{IA-RCRM,IA-RCRM-Globcom} indicates that, instead of leakage
minimization, the conditions of (\ref{DoF-cond-1}) could be dealt
with directly as a rank problem as following:

\normalsize
\begin{equation}\label{DoF-cond-2}
\begin{aligned}
&\mathrm{rank}(\mathcal{J}_j(\{\mathrm{V}_i\}_{i\neq j}^{M},\{\mathrm{G}_{l}\}_{l=1}^L,\mathrm{U}_j))=0  \\
&\mathrm{rank}(\mathcal{S}_j(\mathrm{V}_j,\{\mathrm{G}_{l}\}_{l=1}^L,\mathrm{U}_j))=d
\end{aligned}
\end{equation}
\normalsize

Then the problem (\ref{DoF-cond-2}) is further formulated as a rank
constrained rank minimization as following:

\normalsize
\begin{equation}\label{RCRM}
\begin{aligned}
&
{\text{minimize}}
& & \sum_{j=1}^{K}\mathrm{rank}(\mathcal{J}_j) \\
& \text{s.t. :} & & \mathrm{rank}(\mathcal{S}_j)=d, \quad \forall
j\in\mathcal{K}
\end{aligned}
\end{equation}
\normalsize

The problem (\ref{RCRM}) guarantees that the interference subspaces
collapse to the smallest possible dimension, under the constraints
of the desired signal subspaces to have full ranks. While minimizing
the sum of ranks of the interference matrices is equivalent to
maximizing the sum desired DoF of the network. 


%

Since the objection function in (\ref{RCRM}) is nonconvex,
\cite{IA-RCRM} proposes a tight convex approximation regarding
$\mathcal{J}_j$:

\normalsize
\begin{equation}\label{convx-approx-1}
\begin{aligned}
\overline{\mathrm{conv}}(\sum_{j=1}^{K}\mathrm{rank}(\mathcal{J}_j))&=\frac{1}{\mu}\sum_{j=1}^{K}\|\mathcal{J}_j\|_* \\
&=\frac{1}{\mu}\sum_{j=1}^{K}\sum_{n=1}^{d}\sigma_n(\mathcal{J}_j)
\end{aligned}
\end{equation}
\normalsize

where $\overline{\mathrm{conv}}(\cdot)$ denotes convex envelope of a
function, and $\|\cdot\|_*=\sum_{n=1}^d\sigma_n(\cdot)$ is the
nuclear form of a matrix, which accounts for the sum of the largest
$d$ singular values. $\mu$ is a constant satisfying $\mu\geq
\max_{j\in\mathcal{K}}\sigma_1(\mathcal{J}_j)$.


Since the constraints in (\ref{RCRM}) is also non-convex feasible
set, \cite{IA-RCRM} provides an approximation of convex set
regarding $\mathcal{S}_j$:

\small
\begin{equation}\label{convx-approx-2}
\begin{aligned}
&\lambda_{\mathrm{min}}(\mathcal{S}_j)>\epsilon\\
&\mathcal{S}_j \succeq 0_{d\times d}
\end{aligned}
\end{equation}
\normalsize

where $\lambda_{\mathrm{min}}(\mathcal{S}_j)$ is the minimum
eigenvalue of $\mathcal{S}_j$ and $\epsilon$ is a small positive
constant. and $\mathcal{S}_j\succeq 0_{d\times d}$ denotes the
matrix is hermitian positive semidefinite.

With the above convex approximations (\ref{convx-approx-1}) and
(\ref{convx-approx-2}), the non-convex problem in (\ref{RCRM}) is
relaxed to a total convex optimization \cite{convex-boyd}.
(\ref{convx-approx-1}) suppresses interference spaces to the
smallest dimension, while (\ref{convx-approx-2}) guarantees desired
signal subspaces. It is equivalent to maximizing DoF of the network.
Then the original coding design problem is formulated as a rank
constrained rank minimization problem as following:


\normalsize
\begin{equation}\label{convx-approx-1and2}
\begin{aligned}
&
{\text{minimize}}
& & \sum_{j=1}^{K}\|\mathrm{\mathcal {J}}_j\|_* \\
\hspace{12mm}& \text{s.t. :}
& & \lambda_{\mathrm{min}}(\mathcal{S}_j)>\epsilon, \\
&&& \mathcal{S}_j \succeq 0_{d\times d}, \qquad \forall
j\in\mathcal{K}
\end{aligned}
\end{equation}
\normalsize


Notice that the optimization problem of (\ref{convx-approx-1and2})
is convex regarding the variables $\mathrm{\mathcal {J}}_j$ and
$\mathcal{S}_j$. In the general IA problems as in \cite{IA-RCRM}, it
is also convex regarding variables $\mathrm{V}_i$ and
$\mathrm{U}_j$. While for relay networks, it is not still guaranteed
if it is convex for the relay coding variables $\Theta_{ji}$ and
$\mathrm{G}_l$.


\subsection{Design of Novel Relay Coding with RCRM in Time-Extended Single Antenna Channels}

Compared with \cite{IA-RCRM}, hereby we specify three groups of
decision variables $\mathrm{G}_l$, $\mathrm{V}_i$, $\mathrm{U}_i$ in
the optimization problem (\ref{convx-approx-1and2}). In the
following part, (\ref{convx-approx-1and2}) is to be proved to be
convex regarding decision variables $\mathrm{G}_{l}$, as well as
original convexity of $\mathrm{V}_i$ and $\mathrm{U}_i$ already
indicated in \cite{IA-RCRM}. The key of the proof relies on the
affine relationship between new variables and the objective function
and constraints, which are shown in lemma \ref{lemma-convex1-0} and
lemma \ref{lemma-convex2-0} respectively as following.

\begin{Lemma}
The objective function (\ref{convx-approx-1}) is convex regarding
variables $\{\mathrm{G}_l\}_{l=1}^L$ .\label{lemma-convex1-0}
\end{Lemma}

\begin{proof}
\cite[\S3.2.2]{convex-boyd} indicates the operation of
\textit{composition with an affine mapping} preserves convexity or
concavity of functions. Suppose a function $f:
\mathbb{R}^n\rightarrow\mathbb{R}$, and a matrix
$A\in\mathbb{R}^{n\times m}$, and a vector $b\in\mathbb{R}^n$.
Define a function $g: \mathbb{R}^m\rightarrow\mathbb{R}$ by
$g(x)=f(Ax+b)$, with $\mathbf{dom}\,g=\{x|Ax+b\in\mathbf{dom}\,f\}$,
where $x\in\mathbb{R}^m$ and $\mathbf{dom}$ denotes the domain. Then
if $f$ is convex,
so is $g$. 

Extract all $T^2L$ elements from all $\mathrm{G}_l$ as
$\mathcal{G}\in\mathbb{C}^{T^2L\times 1}$, and extract all
$d(K-1)dK$ elements from all $\mathcal{J}_j$ as one vector
$\mathcal{E}\in\mathbb{C}^{d^2(K-1)K\times 1}$. According to
(\ref{Theta}) and (\ref{Def-JS}), we obtain that
$\mathcal{E}=\mathcal{A}\mathcal{G}$, where
$\mathcal{A}\in\mathbb{C}^{d^2(K-1)K\times T^2L}$ is composed of
elements from $\{\mathrm{U}_j\}$, $\{\mathrm{V}_i\}$, $\{f_{jl}\}$,
$\{h_{li}\}$. Since $\mathcal{E}$, $\mathcal{A}$ and $\mathcal{G}$
are in the complex domain, they could be transformed into
$\mathcal{E}_{\mathbb{R}}$, $\mathcal{A}_{\mathbb{R}}$ and
$\mathcal{G}_{\mathbb{R}}$ in the real domain by simple matrix
operations. Then
$\mathcal{E}_{\mathbb{R}}=\mathcal{A}_{\mathbb{R}}\mathcal{G}_{\mathbb{R}}$,
where $\mathcal{G}_{\mathbb{R}}\in\mathbb{C}^{2T^2L\times 1}$,
$\mathcal{E}_{\mathbb{R}}\in\mathbb{C}^{2d^2(K-1)K\times 1}$,
$\mathcal{A}_{\mathbb{R}}\in\mathbb{C}^{2d^2(K-1)K\times 2T^2L}$.
Thus, $\mathcal{E}_{\mathbb{R}}$ is an affine mapping from
$\mathcal{G}_{\mathbb{R}}$. \cite{IA-RCRM} concludes
(\ref{convx-approx-1}) is convex regarding
$\mathcal{E}_{\mathbb{R}}$, therefore (\ref{convx-approx-1}) is also
convex regarding $\mathcal{G}_{\mathbb{R}}$, i.e. regarding
$\{\mathrm{G}_l\}_{l=1}^L$.
%
\end{proof}

\begin{Lemma}
The feasible set (\ref{convx-approx-2}) is convex regarding
variables $\{\mathrm{G}_l\}_{l=1}^L$.\label{lemma-convex2-0}
\end{Lemma}

\begin{proof}
\cite[\S2.3]{convex-boyd} indicates that \textit{affine function}
preserves convexity of sets, or allows us to construct convex sets
from others. $S\subseteq \mathbb{R}^n$ is a convex set. If the
function $f: \mathbb{R}^k\rightarrow\mathbb{R}^n$ is affine, i.e.
$f(x)=Ax+b$, where $A\in\mathbb{R}^{n\times k}$,
$x\in\mathbb{R}^{k\times 1}$ and $b\in\mathbb{R}^n$, then the
inverse image of $S$ under $f$: $f^{-1}(S)=\{x|f(x)\in S\}$, is
convex.

Similar to \textit{lemma} \ref{lemma-convex1-0}, extract elements
from $\{\mathrm{S}_j\}$, $\{\mathrm{G}_l\}$ as
$\mathcal{S}_{\mathbb{R}}$, $\mathcal{G}_{\mathbb{R}}$ in the real
domain. According to (\ref{Theta}) and (\ref{Def-JS}),
$\mathcal{S}_{\mathbb{R}}=\mathcal{B}_{\mathbb{R}}\mathcal{G}_{\mathbb{R}}$,
where $\mathcal{B}_{\mathbb{R}}$ is composed of elements from
$\{\mathrm{U}_j\}$, $\{\mathrm{V}_i\}$, $\{f_{jl}\}$, $\{h_{li}\}$.
\cite{IA-RCRM} concludes that feasible set of
$\mathcal{S}_{\mathbb{R}}$ under constraint (\ref{convx-approx-2})
is convex, so that $\mathcal{G}_{\mathbb{R}}$ feasible set as the
inverse image of $\mathcal{S}_{\mathbb{R}}$ is also convex, i.e.
$\{\mathrm{G}_l\}_{l=1}^L$ feasible set is convex.
\end{proof}

%
%


\subsubsection{Numerical Results for the Single Antenna Network}

Thus \textit{lemma} \ref{lemma-convex1} and \textit{lemma}
\ref{lemma-convex2} prove that (\ref{convx-approx-1and2}) is a
typical convex optimization regarding $\{\mathrm{G}_l\}_{l=1}^{L}$.
So that it is convenient to use the software tool CVX to proceed the
numerical results \cite{cvx}. To compare the effects of precoding at
the edge and relay coding in the middle, 8 different cases are
simulated with configurations shown in Table \ref{table-config}.

\normalsize
\begin{table}[htpb]
\centering
\begin{tabular}{|c|c|cccc|c|c|c|} \hline
Case &  1  &  2  &  3  &  4  &  5  &  6  &  7  &  8   \\
\hline
  K  &  3  &  3  &  3  &  3  &  3  &  4  &  4  &  4   \\
\hline
  L  &  5  &  5  &  5  &  5  &  5  & 5  &  5  &  10    \\
\hline
  T  &  10  &  10  &  10  &  10  &  10  &  10  &  15  &  10   \\
\hline
  d  &  4  &  4  &  4  &  4  &  4  &  2  &  4  &  4    \\
\hline
Process  &  1  &  2  &  3  &  4  &  5  &  2  &  2  &  2  \\
\hline
\end{tabular}\caption{Configurations of different simulation cases for relay
coding}\label{table-config}
\end{table}

In Table \ref{table-config}, five different processes are provided
in the comparison as well. Process 1 is set as a reference, where
the relay coding matrices $\{\mathrm{G}_l\}$, precoding and
zero-forcing matrices $\{\mathrm{V}_i\}$ and $\{\mathrm{U}_j\}$ are
randomly generated. While in process 2, only $\{\mathrm{V}_i\}$ and
$\{\mathrm{U}_j\}$ are randomly generated, relay coding matrices
$\mathrm{G}_l$ are designed by solving the RCRM problem
(\ref{convx-approx-1and2}) with CVX \cite{cvx}. Process 3 is on the
contrary, where $\{\mathrm{G}_l\}$ are randomly generated, while
$\mathrm{V}_i$ and $\mathrm{U}_j$ are solved by using the
conventional leakage minimization algorithm of interference
alignment \cite{approaching-capacity-IA}, on the equivalent channels
$\Theta_{ji}$ as in (\ref{Theta}). Process 4 performs the leakage
minimization algorithm to optimize $\mathrm{V}_i$ and $\mathrm{U}_j$
additionally on the basis of Process 2. Process 5 applies the RCRM
method to optimize relay coding $\{\mathrm{G}_l\}$ additionally on
the basis of Process 3.

\begin{figure}[htpb]
  \centering
    \includegraphics[width=4.5in]{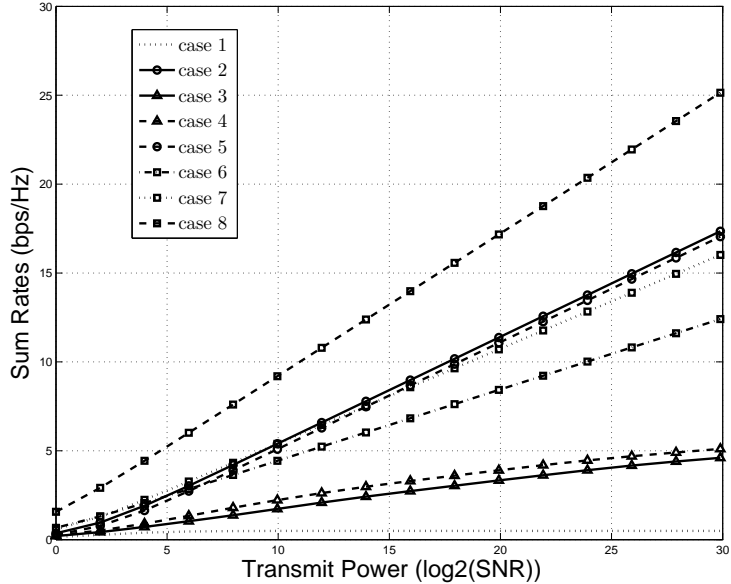}\\
  \caption[Numerical Results of the RCRM method in Time-Extended Single Antenna Channels]{Numerical Results of the RCRM method in Time-Extended Single Antenna Channels (8 cases configured as in Table \ref{table-config})}\label{IA-Relay-Nuclear}
\end{figure}

Numerical results for the above 8 cases in Table \ref{table-config}
are illustrated in Fig. \ref{IA-Relay-Nuclear}. Case 1 shows an
equivalent interference channel as in (\ref{output-HD-bf}) without
any interference alignment design, so that the sum rate deteriorates
in high SNR and the achieved DoF is 0. Case 2 shows that process 2
uses the rank minimization method to design $\{\mathrm{G}_l\}$ to
achieve IA and obtain DoF. Case 3 shows that process 3 uses the
leakage minimization method to design $\mathrm{V}_i$, $\mathrm{U}_j$
to achieves IA as well, but obtaining lower DoF. Case 4 attempts to
additionally optimize $\mathrm{V}_i$, $\mathrm{U}_j$ by leakage
minimization on the basis of process 2, however the DoF is reduced
to the same level of process 3. The reason is the IA structure
formed by relay coding using rank minimization is destroyed, due to
the \textit{inseparable} inter-relationships of relay coding
variables as shown in (\ref{Def-JS}). Case 5 shows that process 5
takes an additional step to optimize $\{\mathrm{G}_l\}$ using rank
minimization method on the basis of process 3. Notice the DoF rises
to the same level of process 2, which means the step of relay coding
plays the key role in achieving interference alignment. Finally,
compare the cases 6, 7 and 8 with case 2. Case 6 shows that when the
number of users increase from 3 to 4, interference alignment is no
longer feasible. So that less concurrent streams are allowed, i.e.
$d$ is reduced from $4$ to $d=2$ to resume IA. Case 7 shows that
when the number of users increase to 4, while maintaining the same
number of streams as case 2, IA could be still feasible if we
increase the time extension length from $T=10$ to $T=15$. However
the DoF is yet not increased along with the number of streams,
because time extension should be considered resulting in
$\text{DoF}=d/T$. So that in case 8, in order to obtain the same
number of DoF per user as in case 2, the issue is settled by
increasing the number of relays from $L=5$ to $L=10$.


%
%
%
%

\subsection{Applying Novel Relay Coding to Diverse General MIMO Relay Channels}

The novel application of rank constrained rank minimization (RCRM)
method is proved to be very effective in the relay coding of the
$K$-pair time-extended single antenna networks. Furthermore, it is
natural to think of extending this idea to more general and
realistic scenarios. Network with all nodes equipped with multiple
antennas (i.e. MIMO networks) are usually investigated in other
works
\cite{IA-MIMO-AF-IFC,relay-beamforming-Interference-Pricing-two-hop}.
It is in general hard to design beamforming to mitigate and align
interference with conventional optimization techniques. For example,
the end-to-end sum-rate maximization problem is non-convex and
NP-hard, and implicit MSE(minimum square error) metric and the
interference pricing method are considered to obtain sub-optimal
solutions
\cite{IA-MIMO-AF-IFC,relay-beamforming-Interference-Pricing-two-hop}.
RCRM method is a promising efficient approach to deal with such MIMO
interference channels \cite{IA-RCRM,IA-RCRM-Globcom}. It is yet
non-trivial to apply the RCRM method to MIMO interference channels
with relays, where the convexity for relay coding is an issue.
Anyhow, the novel approach to apply RCRM in MIMO relay networks
provides an simple accessible way.

The following work focuses on general MIMO relay networks without
any time symbol extensions. Although it is usually infeasible for
MIMO networks to obtain 1/2 DoF for every user by interference
alignment \cite{feasibility-IA-MIMO-IFC}, the idea of suppressing
interference subspace dimension could be still used to obtain high
DoF, which captures the key value of interference alignment at high
SNR. While RCRM method provides such a measure to minimize the
subspace dimension to obtain suboptimal DoF through convex
approximation \cite{IA-RCRM,IA-RCRM-Globcom}. The RCRM method is
non-trivially extended in this work to apply to relay-connected MIMO
networks. So that relay coding is involved and designed along with
precoding at source and destination nodes, which reflects the
strategy of \textit{coding in the middle}
\cite{NC-Three-Unicast-IA}.

Moreover, this novel design is further found to be naturally
applicable to a wide range of typical relay networks, which are
originally hard to analyze and solve. These scenarios include MIMO
relay channel, MIMO two-way relay channel, MIMO Y channel, and MIMO
multi-hop relay channel.

\subsubsection{MIMO Relay Channel}

MIMO relay channel exactly means a two-hop interference channel with
half-duplex relays and a two-stage transmission procedure. Consider
the model in Fig. \ref{relay-IA-pic2}. There are $K$ sources denoted
from $I_1$ to $I_K$; $K$ destinations denoted from $J_1$ to $J_K$;
$L$ relays denoted from $R_1$ to $R_L$. All nodes are set to be
equipped with multiple antennas. $I_i$ has $M_{\xi}$ antennas,
$\forall i\in\mathcal{K}$; $R_l$ has $M_{\chi}$ antennas, $\forall
l\in\mathcal{L}$; $J_j$ has $M_{\omega}$ antennas, $\forall
j\in\mathcal{K}$. All the channels are constant.
$\mathrm{H}_{R_lI_i}$ and $\mathrm{H}_{J_jR_l}$ as channels from
$I_i$ to $R_l$ and from $R_l$ to $J_j$ respectively.
$\mathrm{H}_{R_lI_i}\in\mathbb{C}^{M_{\chi}\times M_{\xi}}$ and
$\mathrm{H}_{J_jR_l}\in\mathbb{C}^{M_{\omega}\times M_{\chi}}$. Then
the equivalent channel $\Theta_{ji}\in\mathbb{C}^{M_{\omega}\times
M_{\xi}}$ is expressed as:

\normalsize
\begin{equation}\label{Theta-RCRM}
\begin{aligned}
\Theta_{ji}=\sum_{l=1}^{L}\mathrm{H}_{J_jR_l}\mathrm{G}_{l}\mathrm{H}_{R_lI_i}
\end{aligned}
\end{equation}
\normalsize

For each user $k$, $I_k$ transmits $d$ streams to $J_k$. So that set
$\mathrm{V}_{i}\in\mathbb{C}^{M_{\xi}\times d}$ as the precoding
matrix at source $I_i$; $\mathrm{G}_{l}\in\mathbb{C}^{M_{\chi}\times
M_{\chi}}$ as the network coding matrix at relay $R_l$;
$\mathrm{U}_{j}\in\mathbb{C}^{M_{\omega}\times d}$ as the receive
filter matrix at destination $J_j$.

The procedure to apply RCRM method in this MIMO relay network is the
same as above in the time-extended single antenna network. The
objective function and feasible set are in the same expressions of
(\ref{convx-approx-1}) and (\ref{convx-approx-2}). The final problem
is also in the expression of (\ref{convx-approx-1and2}). It is
necessary to prove the convexity of the problem in MIMO relay
networks by following lemma \ref{lemma-convex1} and lemma
\ref{lemma-convex2}.

\begin{Lemma}
The objective function (\ref{convx-approx-1}) for MIMO relay
networks is convex regarding
$\{\mathrm{G}_l\}_{l=1}^L$.\label{lemma-convex1}
\end{Lemma}

\begin{proof}
\cite[\S3.2.2]{convex-boyd} indicates the operation of
\textit{composition with an affine mapping} preserves convexity or
concavity of functions. Suppose $f:
\mathbb{R}^n\rightarrow\mathbb{R}$, $A\in\mathbb{R}^{n\times m}$,
and $b\in\mathbb{R}^n$. Define $g:
\mathbb{R}^m\rightarrow\mathbb{R}$ as $g(x)=f(Ax+b)$, with
$\mathbf{dom}g=\{x|Ax+b\in\mathbf{dom}f\}$. Then if $f$ is convex,
so is $g$. 

Extract all ${M_{\chi}}^2L$ elements from all $\mathrm{G}_l$ as
$\mathcal{G}\in\mathbb{C}^{{M_{\chi}}^2L\times 1}$; extract all
$d(K-1)dK$ elements from all $\mathcal{J}_j$ as
$\mathcal{F}\in\mathbb{C}^{d^2(K-1)K\times 1}$. According to
(\ref{Theta-RCRM}) and (\ref{Def-JS}), obtain that
$\mathcal{F}=\mathcal{W}\mathcal{G}$, where
$\mathcal{W}\in\mathbb{C}^{d^2(K-1)K\times {M_{\chi}}^2L}$ is
composed of elements from $\{\mathrm{U}_j\}$, $\{\mathrm{V}_i\}$,
$\{\mathrm{H}_{J_jR_l}\}$, $\{\mathrm{H}_{R_lI_i}\}$. Express
$\mathcal{F}$, $\mathcal{W}$, $\mathcal{G}$ in complex domain to
equivalent $\mathcal{F}_{\mathbb{R}}$, $\mathcal{W}_{\mathbb{R}}$
(transformation via matrix operations), $\mathcal{G}_{\mathbb{R}}$
in real domain, so that
$\mathcal{F}_{\mathbb{R}}=\mathcal{W}_{\mathbb{R}}\mathcal{G}_{\mathbb{R}}$,
where $\mathcal{G}_{\mathbb{R}}\in\mathbb{C}^{2{M_{\chi}}^2L\times
1}$, $\mathcal{F}_{\mathbb{R}}\in\mathbb{C}^{2d^2(K-1)K\times 1}$,
$\mathcal{W}_{\mathbb{R}}\in\mathbb{C}^{2d^2(K-1)K\times
2{M_{\chi}}^2L}$. So $\mathcal{F}_{\mathbb{R}}$ is an affine mapping
from $\mathcal{G}_{\mathbb{R}}$. \cite{IA-RCRM} concludes
(\ref{convx-approx-1}) is convex regarding
$\mathcal{F}_{\mathbb{R}}$, therefore (\ref{convx-approx-1}) is also
convex regarding $\mathcal{G}_{\mathbb{R}}$, i.e. regarding
$\{\mathrm{G}_l\}_{l=1}^L$.
\end{proof}

\begin{Lemma}
The feasible set (\ref{convx-approx-2}) for MIMO relay networks is
convex regarding $\{\mathrm{G}_l\}_{l=1}^L$.\label{lemma-convex2}
\end{Lemma}

\begin{proof}
\cite[\S2.3]{convex-boyd} indicates that \textit{affine functions}
preserve convexity of sets, or allow us to construct convex sets
from others. $S\subseteq \mathbb{R}^n$ is a convex set. If $f:
\mathbb{R}^k\rightarrow\mathbb{R}^n$ is affine, i.e. $f(x)=Ax+b$,
where $A\in\mathbb{R}^{n\times k}$ and $b\in\mathbb{R}^n$, the
inverse image of $S$ under $f$: $f^{-1}(S)=\{x|f(x)\in S\}$, is
convex.

Similar to \textit{Lemma} \ref{lemma-convex1}, extract elements from
$\{\mathcal{S}_j\}$, $\{\mathrm{G}_l\}$ as
$\mathcal{F}_{\mathbb{R}}$, $\mathcal{G}_{\mathbb{R}}$ in real
domain. According to (\ref{Theta-RCRM}) and (\ref{Def-JS}),
$\mathcal{F}_{\mathbb{R}}=\mathcal{W}_{\mathbb{R}}\mathcal{G}_{\mathbb{R}}$,
where $\mathcal{F}_{\mathbb{R}}$ is composed of elements from
$\{\mathrm{U}_j\}$, $\{\mathrm{V}_i\}$, $\{\mathrm{H}_{J_jR_l}\}$,
$\{H_{R_lI_i}\}$. \cite{IA-RCRM} concludes that feasible set of
$\mathcal{F}_{\mathbb{R}}$ under constraint (\ref{convx-approx-2})
is convex, so that the feasible set of $\mathcal{G}_{\mathbb{R}}$ as
the inverse image of $\mathcal{S}_{\mathbb{R}}$ is also convex, i.e.
the feasible set of $\{\mathrm{G}_l\}_{l=1}^L$ is convex.

%
\end{proof}


\textbf{Iterative RCRM algorithm for MIMO relay networks}: Since
\textit{lemma} \ref{lemma-convex1} and \textit{lemma}
\ref{lemma-convex2} prove that (\ref{convx-approx-1and2})
 is a typical convex optimization regarding
$\{\mathrm{G}_l\}$ for MIMO relay network, as well as regarding
$\{\mathrm{U}_j\}$ and $\{\mathrm{V}_i\}$. Then the optimization
problem could be proceeded in an iterative manner, regarding
$\{\mathrm{V}_i\}$, $\{\mathrm{G}_l\}$, $\{\mathrm{U}_j\}$, until
convergence. At every step, the objective function
(\ref{convx-approx-1}) is non-increasing so that convergence is
guaranteed. Compared with the algorithms relating to mean squared
error (MSE) and sum rate as in \cite{IA-MIMO-AF-IFC}, this algorithm
directly optimizes DoF for high SNR for the MIMO relay network.

\subsubsection{MIMO Two-Way Relay Channel}

%
%
%
%

%

Consider the $K$-pair two-way (bidirectional) relay MIMO
interference channel as in \cite{NC-IA-K-Bidirection-Relay}. The
model is also as shown in Fig. \ref{relay-IA-pic2}. All $I_k$ and
$J_k$ nodes are both sources and destinations, i.e. each $I_k$
exchanges data with its partner node $J_k$ via intermediate relaying
nodes $R_1,R_2,\ldots,R_L$. The whole procedure are three steps:
uplink, forward, downlink. In uplink step, all end nodes $I_k$ and
$J_k$ simultaneously transmit data towards relay nodes
$R_1,R_2,\ldots,R_L$; in forward step, relay node $R_l$ filters its
received signal $\mathrm{Y}_{R_l}$ through a forwarding filter
$\mathrm{G}_l$; in downlink step, relay nodes $R_1,R_2,\ldots,R_L$
broadcast network-coded combined signal to all end nodes $I_k$ and
$J_k$. Denote $\mathrm{H}_{R_lI_i}$, $\mathrm{H}_{I_iR_l}$,
$\mathrm{H}_{J_jR_l}$ and $\mathrm{H}_{R_lJ_j}$ as channels from
$I_i$ to $R_l$, $R_l$ to $I_i$, $R_l$ to $J_j$ and $J_j$ to $R_l$,
respectively. $\mathrm{H}_{R_lI_i}\in\mathbb{C}^{M_{\chi}\times
M_{\xi}}$, $\mathrm{H}_{I_iR_l}\in\mathbb{C}^{M_{\xi}\times
M_{\chi}}$, $\mathrm{H}_{J_jR_l}\in\mathbb{C}^{M_{\omega}\times
M_{\chi}}$, $\mathrm{H}_{R_lJ_j}\in\mathbb{C}^{M_{\chi}\times
M_{\omega}}$. Although it does not affect our scheme, notice the
downlink channel is reciprocal to the uplink channel, i.e.
$\mathrm{H}_{R_lI_i}=\mathrm{H}_{I_iR_l}^{\dag}$,
$\mathrm{H}_{J_jR_l}=\mathrm{H}_{R_lJ_j}^{\dag}$.

Then the received signals $\mathrm{Y}_{J_k}$ and $\mathrm{Y}_{I_k}$
at $J_k$ and $I_k$ after self interference cancellation (also ignore
noise term) are presented as:

\normalsize
\begin{equation}\label{Two-Way-Yj}
\begin{aligned}
\mathrm{Y}_{J_k}&=\sum_{l=1}^L\mathrm{H}_{J_kR_l}\mathrm{G}_l\big(\sum_{i=1}^K\mathrm{H}_{R_lI_i}\mathrm{V}_{I_i}\mathrm{X}_{I_i}+\sum_{j=1,j\neq k}^K\mathrm{H}_{R_lJ_j}\mathrm{V}_{J_j}\mathrm{X}_{J_j}\big)\\
&=\Phi_{J_kI_k}\mathrm{V}_{I_k}\mathrm{X}_{I_k}+\sum_{i=1,i\neq k}^K\Phi_{J_kI_i}\mathrm{V}_{I_i}\mathrm{X}_{I_i}+\sum_{j=1,j\neq k}^K\Phi_{J_kJ_j}\mathrm{V}_{J_j}\mathrm{X}_{J_j}\\
\text{where }&\ \ \Phi_{J_kI_i}=\sum_{l=1}^L\mathrm{H}_{J_kR_l}\mathrm{G}_l\mathrm{H}_{R_lI_i}\ \ \text{and}\ \ \Phi_{J_kJ_j}=\sum_{l=1}^L\mathrm{H}_{J_kR_l}\mathrm{G}_l\mathrm{H}_{R_lJ_j}\\
\mathrm{Y}_{I_k}&=\sum_{l=1}^L\mathrm{H}_{I_kR_l}\mathrm{G}_l\big(\sum_{i=1,i\neq k}^K\mathrm{H}_{R_lI_i}\mathrm{V}_{I_i}\mathrm{X}_{I_i}+\sum_{j=1}^K\mathrm{H}_{R_lJ_j}\mathrm{V}_{J_j}\mathrm{X}_{J_j}\big)\\
&=\Phi_{I_kJ_k}\mathrm{V}_{J_k}\mathrm{X}_{J_k}+\sum_{i=1,i\neq k}^K\Phi_{I_kI_i}\mathrm{V}_{I_i}\mathrm{X}_{I_i}+\sum_{j=1,j\neq k}^K\Phi_{I_kJ_j}\mathrm{V}_{J_j}\mathrm{X}_{J_j}\\
\text{where }&\ \
\Phi_{I_kI_i}=\sum_{l=1}^L\mathrm{H}_{I_kR_l}\mathrm{G}_l\mathrm{H}_{R_lI_i}\
\ \text{and}\ \ \
\Phi_{I_kJ_j}=\sum_{l=1}^L\mathrm{H}_{I_kR_l}\mathrm{G}_l\mathrm{H}_{R_lJ_j}
\end{aligned}
\end{equation}
\normalsize

(\ref{Two-Way-Yj}) shows the model is equivalent to $2K$ users
transmitting in pairs, and there are $(2K)^2$ equivalent channels
between transmitters and receivers which are $\{\Phi_{J_kI_i}\}$,
$\{\Phi_{J_kJ_j}\}$, $\{\Phi_{I_kI_i}\}$, $\{\Phi_{I_kJ_j}\}$. The
coding procedure is just the same as the case of MIMO relay
networks. We could obtain the same RCRM problem for MIMO two-way
relay networks as in (\ref{convx-approx-1and2}), which is proved to
be a convex optimization through the following similar lemmas.

\begin{Lemma}\label{lemma-convex1-two-way}
For MIMO two-way relay networks as in (\ref{Two-Way-Yj}), the
corresponding rank minimization objective function
(\ref{convx-approx-1}) is convex regarding
$\{\mathrm{G}_l\}_{l=1}^L$ .
\end{Lemma}
\begin{proof}
Proof is similar to lemma \ref{lemma-convex1}.
\end{proof}

\begin{Lemma}\label{lemma-convex2-two-way}
For MIMO two-way relay networks as in (\ref{Two-Way-Yj}), the
corresponding rank constrained feasible set (\ref{convx-approx-2})
is convex regarding $\{\mathrm{G}_l\}_{l=1}^L$.
\end{Lemma}
\begin{proof}
Proof is similar to lemma \ref{lemma-convex2}.
\end{proof}

\subsubsection{MIMO Y Channel}

Consider MIMO Y channel as in
\cite{DoF-MIMO-Y-signal-space-alignment-NC}.
The model is also illustrated by Fig. \ref{relay-IA-pic2}. The
network only has three user nodes which are set as $I_1$, $I_2$,
$I_3$, and one relay node set as $R_1$. Each user intends to convey
independent messages for the other two users via the relay, while
receiving two independent messages from the other two users. The
transmission has two phases. First phase is the multiple-access
channel (MAC) phase, when all users transmit to the relay; second
phase is the broadcast (BC) phase, when the relay generates new
transmitting signals and send to all users. Denote
$\mathrm{V}_{I_i}^{[I_k]}$ and $\mathrm{X}_{I_i}^{[I_k]}$ as the
precoding matrix and input data from $I_i$ to $I_k$.




Then the received signals $\mathrm{Y}_{I_k}$ at $I_k$ after self
interference cancellation and $\mathrm{Y}_{R_1}$ at $R_1$ (both
ignore noise terms) are:

\normalsize
\begin{equation}\label{MIMO-Y-YI}
\begin{aligned}
&\mathrm{Y}_{R_1}=\mathrm{H}_{R_1I_1}\big(\mathrm{V}_{I_1}^{[I_2]}\mathrm{X}_{I_1}^{[I_2]}+\mathrm{V}_{I_1}^{[I_3]}\mathrm{X}_{I_1}^{[I_3]}\big)\\
&\hspace{9mm}+\mathrm{H}_{R_1I_2}\big(\mathrm{V}_{I_2}^{[I_3]}\mathrm{X}_{I_2}^{[I_3]}+\mathrm{V}_{I_2}^{[I_1]}\mathrm{X}_{I_2}^{[I_1]}\big)\\
&\hspace{9mm}+\mathrm{H}_{R_1I_3}\big(\mathrm{V}_{I_3}^{[I_1]}\mathrm{X}_{I_3}^{[I_1]}+\mathrm{V}_{I_3}^{[I_1]}\mathrm{X}_{I_3}^{[I_1]}\big)\\
&\mathrm{Y}_{I_k}=\mathrm{H}_{I_kR_1}\mathrm{G}_1\mathrm{Y}_{R_1}\\
&\text{so that: }\\
&\mathrm{Y}_{I_1}=\big(\mathrm{H}_{I_1R_1}\mathrm{G}_1\mathrm{H}_{R_1I_2}\big)\mathrm{V}_{I_2}^{[I_1]}\mathrm{X}_{I_2}^{[I_1]}+\big(\mathrm{H}_{I_1R_1}\mathrm{G}_1\mathrm{H}_{R_1I_3}\big)\mathrm{V}_{I_3}^{[I_1]}\mathrm{X}_{I_3}^{[I_1]}\\
&\mathrm{Y}_{I_2}=\big(\mathrm{H}_{I_2R_1}\mathrm{G}_1\mathrm{H}_{R_1I_3}\big)\mathrm{V}_{I_3}^{[I_2]}\mathrm{X}_{I_3}^{[I_2]}+\big(\mathrm{H}_{I_2R_1}\mathrm{G}_1\mathrm{H}_{R_1I_1}\big)\mathrm{V}_{I_1}^{[I_2]}\mathrm{X}_{I_1}^{[I_2]}\\
&\mathrm{Y}_{I_3}=\big(\mathrm{H}_{I_3R_1}\mathrm{G}_1\mathrm{H}_{R_1I_1}\big)\mathrm{V}_{I_1}^{[I_3]}\mathrm{X}_{I_1}^{[I_3]}+\big(\mathrm{H}_{I_3R_1}\mathrm{G}_1\mathrm{H}_{R_1I_2}\big)\mathrm{V}_{I_2}^{[I_3]}\mathrm{X}_{I_2}^{[I_3]}\\
\end{aligned}
\end{equation}
\normalsize

(\ref{MIMO-Y-YI}) shows the model is equivalent to $6$ pairs of
users transmitting on $36$ equivalent channels. We could obtain the
same RCRM problem as in (\ref{convx-approx-1and2}) for MIMO Y
channel. It is proved to be convex through the following similar
lemmas.

\begin{Lemma}\label{lemma-convex1-MIMO-Y}
For the MIMO Y channel as in (\ref{MIMO-Y-YI}), the corresponding
rank minimization objective function (\ref{convx-approx-1}) is
convex regarding $\{\mathrm{G}_l\}_{l=1}^L$ .
\end{Lemma}
\begin{proof}
Proof is similar to lemma \ref{lemma-convex1}.
\end{proof}

\begin{Lemma}\label{lemma-convex2-MIMO-Y}
For the MIMO Y channel as in (\ref{MIMO-Y-YI}), the corresponding
rank constrained feasible set (\ref{convx-approx-2}) is convex
regarding $\{\mathrm{G}_l\}_{l=1}^L$.
\end{Lemma}
\begin{proof}
Proof is similar to lemma \ref{lemma-convex2}.
\end{proof}

\hspace{1mm} Compared with
\cite{DoF-MIMO-Y-signal-space-alignment-NC}, the analytical signal
space alignment and network-coding-aware interference nulling
beamforming scheme is complicated to be implemented and hard to
adapt to changes of structure. While this RCRM method is a direct
solution to form efficient and robust coding.

\subsubsection{MIMO Multi-hop Relay Channel}

Consider MIMO interference channels with multi-hop relays as in
\cite{DOF-Multisource-Relay,DoF-region-class-multisource-Gaussian-relay}.
The model in Fig. \ref{relay-IA-pic2} needs to be modified. There
are $M$ layers of parallel relaying nodes between sources and
destinations in the network. Denote the $m$-th layer's nodes as
$R_{1,m},R_{2,m},\ldots,R_{L,m}$, and the corresponding filter
matrices are
$\mathrm{G}_{1,m},\mathrm{G}_{2,m},\ldots,\mathrm{G}_{L,m}$. The
channels are denoted as $\mathrm{H}_{J_jR_{l,M}}$ from $M$-th layer
of relays to destination $J_j$, $\mathrm{H}_{R_{l,m}R_{l,m-1}}$ from
$(m-1)$-th layer of relays to $m$-th layer of relays,
$\mathrm{H}_{R_{l,1}I_i}$ from source $I_i$ to $1$-st layer of
relays.


For simplicity, only assume there are two layers of relays.
$\mathrm{G}_{q,1}$ is the filter at $q$-th relay in first layer and
$\mathrm{G}_{p,2}$ is the filter at $p$-th relay in second layer.
$\mathrm{H}_{R_{q,1}I_i}$, $\mathrm{H}_{R_{p,2}R_{q,1}}$, and
$\mathrm{H}_{J_jR_{p,2}}$ are channels from source $I_i$ to relay
$R_{q,1}$, from relay $R_{q,1}$ to relay $R_{p,2}$, and from relay
$R_{p,2}$ to destination $J_j$. Then the received signal
$\mathrm{Y}_{J_j}$ at $J_j$ is:

\normalsize
\begin{equation}\label{Multi-hop-Yj}
\begin{aligned}
\mathrm{Y}_{J_j}&=\sum_{p=1}^L\mathrm{H}_{J_jR_{p,2}}\mathrm{G}_{p,2}\sum_{q=1}^L\mathrm{H}_{R_{p,2}R_{q,1}}\mathrm{G}_{q,1}\sum_{i=1}^K\mathrm{H}_{R_{q,1}I_i}\mathrm{V}_{I_i}\mathrm{X}_{I_i}\\
\end{aligned}
\end{equation}
\normalsize

The equivalent channel from $I_i$ to $J_j$ is:

$\Phi_{J_jI_i}=\sum_{p=1}^L\mathrm{H}_{J_jR_{p,2}}\mathrm{G}_{p,2}\sum_{q=1}^L\mathrm{H}_{R_{p,2}R_{q,1}}\mathrm{G}_{q,1}\mathrm{H}_{R_{q,1}I_i}$

\hspace{1mm}

\textbf{Iterative algorithms layer by layer:} At each step, choose
one layer of relays as the coding variables, then we could obtain
the same RCRM problem as in (\ref{convx-approx-1and2}). The whole
algorithm is convergent, since the objective function maintains the
same all the time. It is proved to be a convex optimization problem
through following similar lemmas.

\begin{Lemma}\label{lemma-convex1-multi-hop}
For the MIMO interference channels with multi-hop relays as in
(\ref{Multi-hop-Yj}), the corresponding rank minimization objective
function (\ref{convx-approx-1}) is convex regarding
$\{\mathrm{G}_{l,m}\}_{(l,m)=(1,1)}^{(L,M)}$.
\end{Lemma}
\begin{proof}
Proof is similar to lemma \ref{lemma-convex1}.
\end{proof}

\begin{Lemma}\label{lemma-convex2-multi-hop}
For the MIMO interference channels with multi-hop relays as in
(\ref{Multi-hop-Yj}), the corresponding rank constrained feasible
set (\ref{convx-approx-2}) is convex regarding
$\{\mathrm{G}_{l,m}\}_{(l,m)=(1,1)}^{(L,M)}$.
\end{Lemma}
\begin{proof}
Proof is similar to lemma \ref{lemma-convex2}.
\end{proof}

\hspace{1mm}

In summary, MIMO interference channels with multi-hop relays are
generally hard to have either analytical or numerical solutions
\cite{DOF-Multisource-Relay,DoF-region-class-multisource-Gaussian-relay}.
While this novel RCRM based algorithm provides a robust and
universal solution.

\subsubsection{Numerical Results for Diverse MIMO Relay Networks}

All the numerical results are shown in Fig.
\ref{relay-IA-4case-ISIT}. The above four scenarios of MIMO relay
channel, MIMO two-way relay channel, MIMO Y channel and MIMO
multi-hop channel are denoted as case 1, case 2, case 3, and case 4,
respectively. Configure the four cases as shown in Table
\ref{table-config-ISIT}. In case 3, '1$\times$2' means 2 separate
streams at one user in the MIMO Y channel. In case 4, '3$\times$2'
means 2 layers of relays (3 relay in each layer) in the multi-hop
relay network.

\footnotesize
\begin{table}[!hbp]
\centering
\begin{tabular}{|c|c|c|c|c|c|c|} \hline
Case &  $K$  &  $L$  &  $M_{\xi}$  &  $M_{\chi}$  &  $M_{\omega}$  &  d    \\
\hline
  1  &  3  &  3  &  10  &  5  &  10  &  2    \\
\hline
  2  &  3  &  3  &  10  &  10  &  10  & 2    \\
\hline
  3  &  3  &  1  &  10  &  10  &  10  &  1$\times$2    \\
\hline
  4  &  3  &  3$\times$2  &  10  &  5  &  10  &  2     \\
\hline
\end{tabular}
\caption{Configurations of 4 cases of MIMO relay
networks}\label{table-config-ISIT}
\end{table} \normalsize

Simulation has been done with the software tool CVX \cite{cvx}. In
Fig. \ref{relay-IA-4case-ISIT}, 'Ref' means random coding as a
reference for the 4 kinds of MIMO relay networks, which obviously
obtain 0 DoF. While 'Opt' means results with the IA coding design,
which shows that DoF is well obtained. So that the novel application
of RCRM method in relay networks successfully implements
interference alignment in typical scenarios. Case 2 has the largest
DoF because it has bidirectional transmissions. In case 3, MIMO Y
channel transmits 2 datastreams for each user, however there is only
1 DoF for each virtual port. In case 4, multi-hop relays relaxe the
requirements for individual relays in each layer to maitain high
DoF. In summary, this novel relay coding design successfully
implements interference alignment in above relay-involved scenarios
as a effective universal solution.

\begin{figure}[htpb]
  \begin{center}
    \includegraphics[width=4.5in]{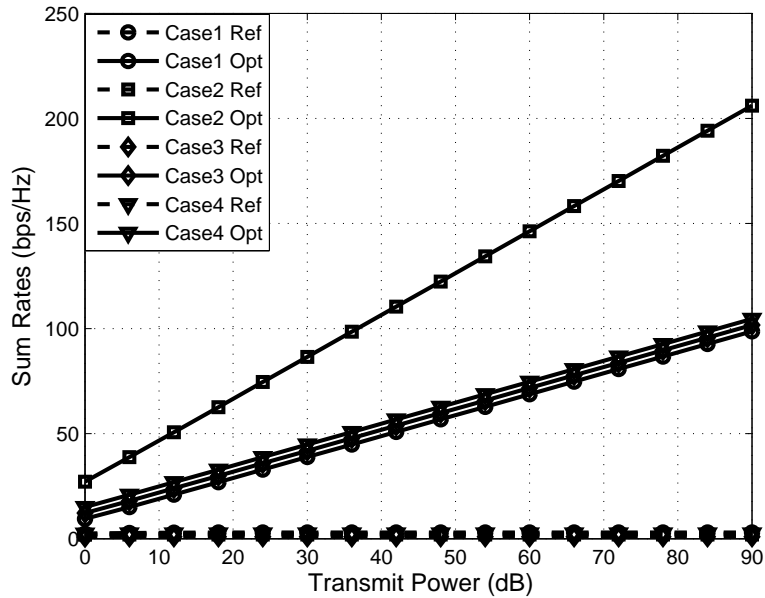}\\
  \end{center}
  \caption[Rates of 4 cases of MIMO relay networks]{Rates of 4 cases of MIMO relay networks (4 cases as configured in Table
\ref{table-config-ISIT} before and after relay coding
optimization)}\label{relay-IA-4case-ISIT}
\end{figure}

\section{Conclusion}
\label{sect:conc} This work considers two strategies in which a
finite number of relays can be employed to achieve IA in a $K$-pair
interference network. In the coding-at-the-edge strategy, in the
single antenna constant channel, relays can be used with time
extension techniques to produce equivalent channels where standard
IA schemes are applicable. However, at least two relays are proved
to be required to produce channel randomness, which is a necessity
for all kinds of IA schemes to be achievable. Novel precoding
schemes with double-layered symbol extensions could also be used to
investigate the issue of constant channels.
In the coding-in-the-middle strategy, as
a universal numerical approach for diverse general MIMO networks
with relays, relay coding can be designed as an optimization problem
to minimize the rank of the interference subspace. Efficient and
robust algorithms are proposed.

\bibliographystyle{IEEEtran}
\bibliography{Thesis}

\end{document}